\documentclass[conference]{IEEEtran}
\IEEEoverridecommandlockouts

\usepackage{cite}
\usepackage{amsmath,amssymb,amsfonts,amsthm}
\usepackage{algorithmic}
\usepackage{algorithm}
\usepackage{graphicx}
\usepackage{textcomp}
\usepackage{xcolor}
\usepackage{booktabs}
\usepackage{multirow}
\usepackage{enumitem}
\usepackage{bm}
\usepackage{tikz}
\usepackage{url}
\usetikzlibrary{arrows.meta, positioning}

\newcommand{\E}{\mathbb{E}}

\newcommand{\calP}{\mathcal{P}}

\newcommand{\calI}{\mathcal{I}}

\newcommand{\calK}{\mathcal{K}}
\newcommand{\calZ}{\mathcal{Z}}

\newcommand{\rev}[1]{#1}

\newtheorem{theorem}{Theorem}[section]

\newtheorem{remark}{Remark}

\newtheorem{proposition}[theorem]{Proposition}

\allowdisplaybreaks[4]

\begin{document}

\title{Robust KV Cache Management for LLM Serving under Output Token Length Uncertainty}

\author{ \IEEEauthorblockN{Jiaming Cheng, Duong The Do, Duong Tung Nguyen\\
\IEEEauthorblockA{Arizona State University, Tempe, AZ 85281, USA,
\{jiaming, dtdo4, duongnt\}@asu.edu} \\}}
\maketitle

\begin{abstract}
KV cache memory is a primary bottleneck in modern LLM serving systems deployed on GPU clusters. A fundamental challenge is that the KV cache must be reserved upon request arrival, while the output token length remains unknown until generation completes. Under-reservation triggers preemption---forcing termination and recomputation of requests and incurring significant overhead---whereas over-reservation wastes memory and reduces throughput. This creates a central trade-off between memory efficiency and preemption risk. We present a robust KV cache management framework for LLM serving that jointly optimizes GPU parallelism configuration, KV cache reservation per request class, request routing across heterogeneous serving groups, and prefix caching for shared prompts. The framework incorporates latency SLO constraints and captures the interaction between memory allocation, throughput, and queueing delay. To address output token length uncertainty and workload distribution shift, we develop a Wasserstein distributionally robust optimization (DRO) formulation together with a scalable block coordinate descent algorithm for the resulting mixed-integer problem. Our analysis reveals a critical fractile structure that automatically adapts reservation quantiles to different preemption and memory cost regimes without manual tuning. Trace-driven evaluation on production LLM workloads, including BurstGPT, Azure, and ShareGPT traces, demonstrates up to 56\% lower cost than fixed-quantile reservation baselines while maintaining competitive P99 latency, goodput, and SLO violation rates across diverse operating regimes.
\end{abstract}


%==============================================================================
\section{Introduction}
\label{sec:intro}
%==============================================================================

Cloud-hosted large language model (LLM) inference has become a dominant workload on modern GPU clusters, with batched autoregressive decoding driving aggregate utilization at major providers. A central operational challenge in this setting is the management of the per-request key-value (KV) cache, which grows by one entry per generated token and must be allocated at request admission time, before the decoder has emitted any output. For a 70B-parameter model, the KV cache consumes approximately $2.5$\,MB per token, so a 2{,}000-token response requires $5$\,GB and may exceed the memory of a single accelerator~\cite{vllm2023}. Under-reservation triggers mid-generation preemption: the request is evicted from the running batch, completed tokens are discarded, and the request is recomputed from scratch, inflating time-to-first-token (TTFT), inter-token latency (ITL), and tail-percentile SLO violations. Over-reservation, conversely, leaves KV memory idle, reduces the achievable batch size, and lowers goodput---the rate of admitted-and-on-time requests---below the configured capacity. Because continuous-batching admission in modern serving runtimes is constrained by aggregate KV residency, the per-class reservation decisions directly determine both throughput and service quality. %choice therefore determines both throughput and SLO attainment.

%--- CHALLENGE: Why this problem is hard ---
This problem is challenging for three reasons. \textbf{(C1) Workload heterogeneity.} Output lengths vary by orders of magnitude across application classes: tens of tokens for chat, hundreds for retrieval-augmented question answering, and thousands for code generation. A single global quantile-based reservation cannot therefore be optimal for all classes simultaneously. \textbf{(C2) Non-stationarity.} Workload distributions drift over time as user behavior, prompt templates, and traffic mix evolve; reservation policies fitted to historical traces degrade as the empirical tail shifts. \textbf{(C3) System heterogeneity.} Modern serving stacks partition GPUs into multiple groups under different tensor parallelism (TP) and pipeline parallelism (PP) configurations, employ iteration-level continuous batching, and may disaggregate the prefill and decode phases onto separate engines~\cite{zhong2024distserve,agrawal2024taming}, producing a discrete set of operating points that trade KV capacity, throughput, and tail latency. The cluster operator must therefore jointly determine the GPU partition, the per-class reservation, and the request routing in advance of any per-request output realization.

%--- GAP: What prior work misses ---
Existing work primarily improves LLM serving efficiency through \textit{runtime mechanisms} such as KV-cache paging, continuous batching, and prefill/decode disaggregation~\cite{vllm2023,yu2022orca,zhong2024distserve,agrawal2024taming}. These systems focus on efficient execution once memory allocations are determined, but leave reservation and routing decisions largely to operator-defined heuristics. Other approaches either reduce KV memory through token eviction~\cite{zhang2023h2o,li2024snapkv} or rely on fixed reservation heuristics (Mean, P90, P95, P99, Max), which are sensitive to workload drift and effective only under limited operating regimes. Meanwhile, distributionally robust optimization (DRO)~\cite{mohajerin2018data} has been developed for stochastic decision-making, but has not been applied to KV-cache reservation in LLM serving systems. 
Classical DRO inventory models also do not capture the in-flight memory growth and recomputation overhead specific to LLM serving.
Putting these threads together, no existing work supplies the decision policy that jointly chooses parallelism, reservation, routing, and prefix caching under output-length uncertainty with robustness guarantees. Section~\ref{sec:related} expands this analysis; the gap motivates two research questions:

\smallskip\noindent\textbf{RQ1:} \textit{How should we reserve KV cache under output length uncertainty to minimize the combined cost of preemption (under-reservation) and wasted capacity (over-reservation)?}

\smallskip\noindent\textbf{RQ2:} \textit{What reservation policy is optimal, and how does it adapt to different cost structures without manual tuning?}

%--- APPROACH AND CONTRIBUTIONS ---
To address these challenges, we formulate a joint reservation-routing optimization using distributionally robust optimization (DRO) with Wasserstein ambiguity sets. The single-request reservation reduces to the classical newsvendor problem~\rev{\cite{edgeworth1888,arrow1951}}; what makes our setting distinct is that memory demand grows \emph{during} execution as tokens are generated, under-reservation triggers costly recomputation, and over-reservation reduces concurrent serving capacity---and the reservation is coupled across concurrent requests through a shared GPU-memory pool rather than solved per request in isolation.
Fig.~\ref{fig:system_model} illustrates the proposed system architecture and the interaction between the control and execution layers.

Our analysis reveals that the optimal reservation follows a \emph{critical fractile} structure---the system automatically becomes more conservative when preemption is expensive and more aggressive when over-reservation is costly, eliminating the need for manual quantile tuning.
Our contributions include:
\begin{itemize}[nosep]
    \item We develop a unified optimization framework that jointly coordinates GPU parallelism configuration, KV cache reservation, request routing, and prefix caching under output token length uncertainty and latency SLO constraints. 
    \item We formulate the problem as a Wasserstein DRO model and derive a critical fractile structure for the optimal reservation policy, enabling automatic adaptation to different preemption and memory cost regimes without manual quantile tuning. We design a scalable block coordinate descent algorithm (BCD-DRO) that supports periodic re-optimization under workload drift.
    \item Through trace-driven evaluation on production LLM workloads (BurstGPT, Azure, ShareGPT), we demonstrate that \emph{no single fixed-quantile heuristic is optimal across cost regimes}---each is optimal only in a narrow regime, while DRO automatically adapts to any cost structure.
\end{itemize}

\begin{figure}[t]
\centering
    \includegraphics[width=\columnwidth]{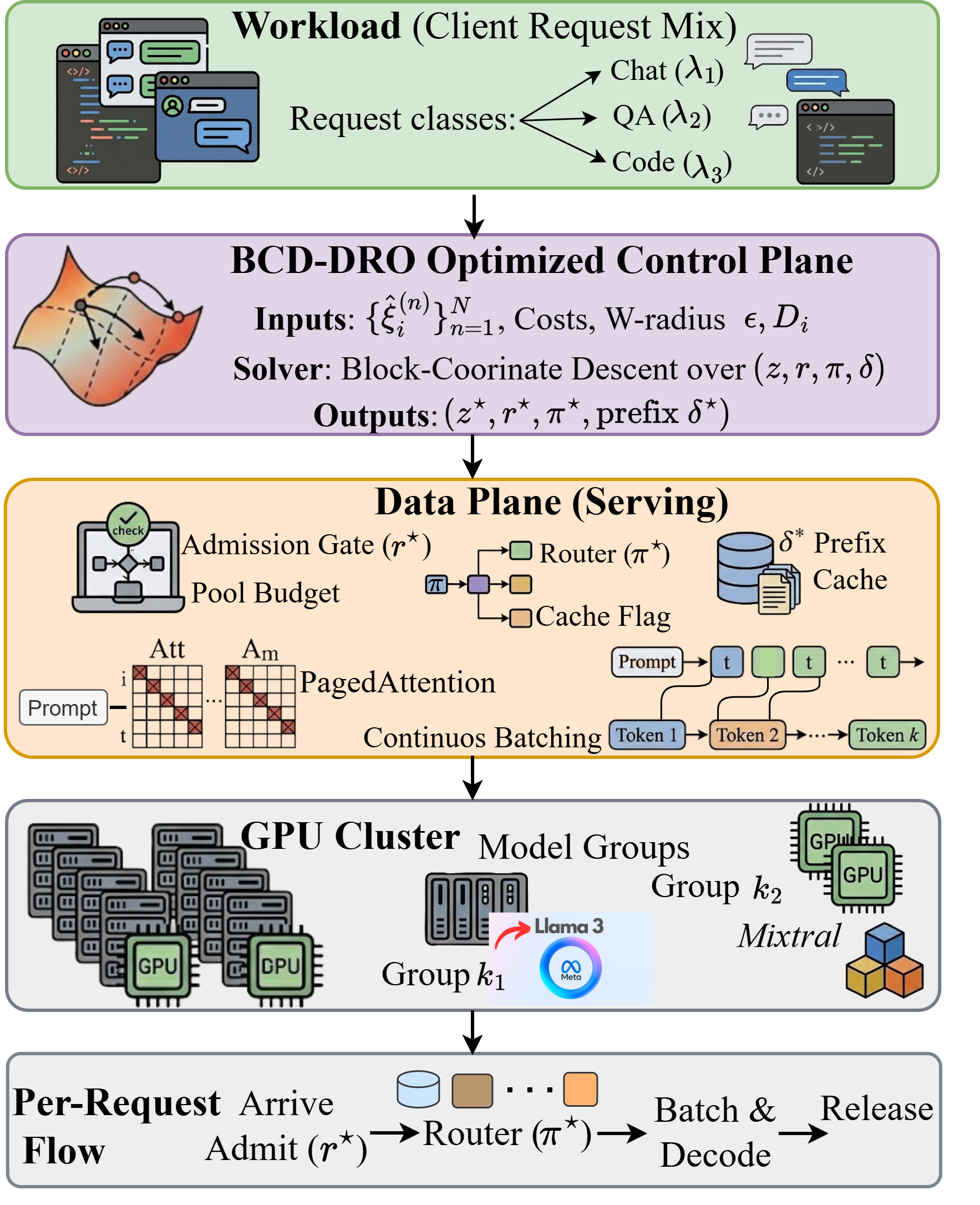}
    \caption{System model: BCD-DRO (purple, this work) is a control-plane policy that ingests historical samples, cost parameters, and a Wasserstein radius and produces four control variables $(z^*, r^*, \pi^*, \delta^*)$ per re-optimization tick. The data plane---PagedAttention, continuous batching, optional P/D split---consumes these variables unchanged; the GPU cluster executes decoding; the dashed arrow closes the rolling-horizon loop with telemetry feedback.}
\label{fig:system_model}
\end{figure}

%==============================================================================
\section{Related Work}
\label{sec:related}
%==============================================================================

We organize prior work into two main categories: runtime KV-cache management and output-length-aware reservation policies. 
While existing work has advanced these directions largely independently, our work integrates them into a unified framework for \textit{uncertainty-aware LLM serving}.

\noindent \textbf{\textit{Runtime KV-Cache Management and Scheduling.}}
Recent LLM serving systems primarily address the KV-cache bottleneck through runtime execution mechanisms rather than reservation policies. PagedAttention~\cite{vllm2023}, the memory management mechanism underlying vLLM, replaces contiguous KV allocation with page-table indirection to improve memory utilization and reduce fragmentation. However, it reacts to memory pressure only \textit{after} KV capacity is exhausted, preempting and later resuming requests \textit{without explicitly determining reservation sizes in advance} or optimizing the KV block pool relative to workload characteristics. Orca-style continuous batching~\cite{yu2022orca} interleaves requests at iteration granularity to maximize GPU utilization while treating per-request memory allocation as an exogenous input. DistServe~\cite{zhong2024distserve} disaggregates prefill and decode execution to improve goodput and tail latency, while Sarathi-Serve~\cite{agrawal2024taming} reduces latency through chunked prefill scheduling. These systems focus on runtime execution once deployment parameters are fixed but leave the key control decisions---per-class reservation, GPU partitioning, request routing, and prefix caching---to the operator. In contrast, our work operates at the decision-policy layer: it jointly determines all four under output-length uncertainty with Wasserstein robustness, and the resulting policies are directly executable by the runtimes above without data-plane modifications (Table~\ref{tab:related_compare}).

\noindent \textbf{\textit{Output-Length-Aware Reservation and Robust Optimization.}}
Another line of work addresses uncertainty in output length and KV-cache growth. Eviction-based methods such as H$_2$O~\cite{zhang2023h2o} and SnapKV~\cite{li2024snapkv} reduce the per-token KV-cache footprint by dropping tokens deemed unimportant, trading accuracy for memory efficiency. These methods are orthogonal to our reservation policy and can be combined with it; improved reservation decisions reduce the amount of KV eviction required and, therefore, mitigate the associated accuracy degradation. Length-prediction approaches~\cite{jin2023s3,zheng2024response,qiu2024power} estimate output length from prompts and reserve memory based on predicted sequence lengths, while fixed-quantile heuristics (Mean, P90, P95, P99, Max) are widely used in practice as deployment defaults in production stacks~\cite{vllm2023}. 

Each of these approaches has limitations. Eviction methods reduce memory usage after execution begins and do not jointly optimize reservation, routing, or configuration. Prediction-based methods are sensitive to estimation error and offer no guarantee against distribution shift. Fixed-quantile heuristics are each optimal only in narrow operating regimes and cannot adapt to changing cost structures. Parallelism strategies~\cite{shoeybi2019megatron,pope2023efficiently} are studied extensively for training, but inference deployments typically fix a configuration and do not jointly optimize memory reservation under uncertainty. Recent workload-allocation and heterogeneous-serving frameworks~\cite{cheng_fastllm,cheng_greenllm} decide \emph{where} to run and \emph{how} to route requests under known demand. In contrast, this work is orthogonal and new: \emph{how much} KV cache to reserve per request when the output length is unknown. This adds the reservation decision, its Wasserstein-DRO formulation and critical-fractile characterization, the capacity constraint that couples reservation to throughput via Little's law, and prefix caching $\delta$.

Distributionally robust optimization (DRO) has been widely studied for decision-making under uncertainty~\cite{mohajerin2018data}, particularly with Wasserstein ambiguity sets for robustness against distribution shift. We bring this machinery into LLM serving by formulating KV-cache reservation as a Wasserstein DRO problem with explicit SLO constraints, jointly co-designing parallelism, reservation, routing, and prefix caching with a critical-fractile reservation structure. Table~\ref{tab:related_compare} summarizes the differences.

\begin{table}[t]
\centering
\caption{Summary of literature}
\label{tab:related_compare}
\scriptsize % Reduced from \footnotesize
\setlength{\tabcolsep}{3pt} % Reduced from 2pt
\begin{tabular}{l|c|c|c|c|c}
\toprule
\textbf{System} & \textbf{Mem.} & \textbf{Res.} & \textbf{Joint} & \textbf{Len.} & \textbf{Robust.} \\
                & \textbf{layout} & \textbf{policy} & \textbf{decis.} & \textbf{model} & \textbf{guar.} \\
\midrule
PagedAttn/vLLM~\cite{vllm2023}    & paged   & reactive & --      & none  & --   \\
Orca~\cite{yu2022orca}            & cont.\ batch & --        & --      & none  & --   \\
DistServe~\cite{zhong2024distserve} & P/D split & --        & topology & none  & --   \\
Sarathi-Serve~\cite{agrawal2024taming} & chunked  & --        & --      & none  & --   \\
H$_2$O~\cite{zhang2023h2o}, SnapKV~\cite{li2024snapkv} & eviction & --        & --      & post-hoc & --   \\
Length predictors                 & any     & point est. & --      & regression & --   \\
\midrule
\textbf{BCD-DRO (ours)}$^*$       & \emph{any} & DRO-opt. & \checkmark & Wasserstein & Wasserstein \\
\bottomrule
\end{tabular}
\vspace{2.5pt} 
\raggedright
\tiny $^*$Only BCD-DRO co-optimizes the four control knobs jointly under an explicit robustness guarantee. ``Mem.'' denotes a runtime memory-management mechanism that operates on inputs provided externally; ``Policy'' denotes a decision rule that produces those inputs.
\end{table}

%==============================================================================
\section{System Model and Problem Formulation}
\label{sec:model}
%==============================================================================

\subsection{System Overview}

We consider an LLM service provider operating a GPU cluster with $J$ GPUs. The provider partitions GPUs into \emph{serving groups}, each using a specific parallelism configuration. Let $\calK$ denote the set of candidate configurations, where each $k \in \calK$ specifies tensor parallelism degree $\tau_k$ and pipeline parallelism degree $p_k$, consuming $\tau_k p_k$ GPUs per group. Each configuration has derived KV cache capacity $M_k$ (in tokens) and effective compute/bandwidth $(C_k, B_k)$.

\smallskip\noindent\textbf{Request Classes.} Requests are partitioned into classes $i \in \calI$ based on application type (chat, code generation, QA). Each class $i$ has: arrival rate $\lambda_i$, input length $L_{\text{in}}^i$ (known), shared prefix $L_{\text{pre}}^i$ (cacheable), latency SLO $D_i$, and output length $\xi_i$ (\emph{uncertain} until generation completes). This class structure is motivated by two observations: (i) requests within each class exhibit \emph{similar output distributions}---enabling learning from historical data, and (ii) requests share \emph{common prefixes}---enabling prefix caching.

\smallskip\noindent\textbf{Service Time.} LLM inference consists of compute-bound prefill and memory-bandwidth-bound decode phases. The expected service time for class $i$ on configuration $k$ is:
\begin{equation}
    T_{i,k} = \underbrace{\frac{p_k \alpha L_{\text{in}}^i}{C_k}}_{\text{prefill + bubble}} + \underbrace{\frac{\beta \E[\xi_i]}{B_k}}_{\text{decode}} + \underbrace{N_L \tau_k t_{\text{ar}}}_{\text{comm.}}
    \label{eq:service_time}
\end{equation}
where $\alpha, \beta$ are model-specific constants, $N_L$ is the number of layers, and $t_{\text{ar}}$ is all-reduce latency. Note that service time uses $\E[\xi_i]$ rather than the uncertain realization $\xi_i$ because throughput is an \emph{aggregate} metric---by the law of large numbers, variability across concurrent requests averages out. In contrast, reservation decisions affect \emph{individual} requests where under-reservation triggers preemption, necessitating the DRO formulation. The offered load is $u_{i,k} = \lambda_i T_{i,k}$.

\subsection{Decision Variables}
\label{sec:formulation}

The operator jointly optimizes four decisions, all made \emph{before} observing actual output lengths:
\begin{itemize}[nosep]
    \item \textbf{Configuration deployment} $z_k \in \mathbb{Z}_+$: number of serving groups using configuration $k \in \calK$. This determines how GPUs are partitioned.
    \item \textbf{Reservation} $r_i \in [L_{\text{in}}^i, \rev{L_{\text{in}}^i + L_{\text{max}}^i}]$: KV cache tokens reserved per class-$i$ request at admission, where $L_{\text{max}}^i$ is the maximum output length.
    \item \textbf{Routing} $\pi_{i,k} \geq 0$: the portion of class-$i$ requests routed to configuration $k$. The rejected portion satisfies $\pi_{i,0} = 1 - \sum_k \pi_{i,k}$.
    \item \textbf{Prefix caching} $\delta_{i,k} \in \{0,1\}$: whether to cache the shared prefix of class $i$ on configuration-$k$ groups, yielding effective reservation $\tilde{r}_{i,k} = r_i - \delta_{i,k} L_{\text{pre}}^i$.
\end{itemize}

\subsection{Cost Structure}

The reservation decision $r_i$ creates a key trade-off:

\smallskip\noindent\textbf{Preemption Cost.} Preemption occurs when the reserved KV cache is insufficient to accommodate the generated tokens, forcing the system to terminate and recompute the request. This event incurs substantial overhead: previously generated tokens are discarded, GPU compute cycles are wasted, and the request must be re-executed, increasing both system load and latency.
Formally, preemption occurs when $L_{\text{in}}^i + \xi_i > r_i$. The resulting cost is proportional to the overflow beyond the reservation:
\begin{equation}
    C_{\text{preempt}} = c_p \cdot (L_{\text{in}}^i + \xi_i - r_i)^+,
\end{equation}
which implies that larger under-reservation leads to more severe recomputation overhead.

\smallskip\noindent\textbf{Waste Cost.} When the reserved KV cache exceeds actual usage, the unused memory cannot be utilized by other requests, effectively reducing system throughput. This opportunity cost is modeled as:
\begin{equation}
    C_{\text{waste}} = c_w \cdot (r_i - L_{\text{in}}^i - \xi_i)^+,
\end{equation}
reflecting the loss of serving capacity due to over-reservation.

\smallskip\noindent\textbf{Cost Ratio.} The preemption and waste terms together capture a fundamental trade-off: reserving too little leads to costly recomputation due to preemption, while reserving too much reduces the number of concurrent requests that can be served. The relative importance of these effects is governed by the cost ratio $\rho = c_p/c_w$, which encodes the asymmetry between under- and over-reservation and reflects the system's tolerance for preemption versus memory inefficiency. This ratio directly determines the optimal reservation quantile, as characterized in Section~\ref{sec:solution}.

\smallskip\noindent\textbf{Total Cost.} The expected cost per class-$i$ arrival combines stochastic and deterministic components:
\begin{equation}
    Q_i = \underbrace{a_i \cdot \E[c_p(\cdot)^+ + c_w(\cdot)^+]}_{\text{preempt/waste (admitted)}} + \underbrace{C_{\text{gpu}} + C_{\text{slo}} + C_{\text{rej}}}_{\text{resource, SLO, rejection}}
    \label{eq:cost}
\end{equation}
where $a_i = \sum_k \pi_{i,k}$ is admission probability, $C_{\text{gpu}} = \sum_k \pi_{i,k} c_{\text{gpu}} \tau_k p_k T_{i,k}$ penalizes GPU usage, $C_{\text{slo}}$ penalizes latency violations, and $C_{\text{rej}} = (1-a_i) c_{\text{rej}}$ penalizes rejections.

\subsection{DRO Formulation}

The output length $\xi_i$ for class $i$ follows an unknown distribution $P^i$. A natural approach is to construct a \emph{class-specific} ambiguity set: requests within the same class share similar application semantics (e.g., chat vs.\ code generation), and thus exhibit similar output length distributions~\cite{zheng2023judging,agrawal2024taming,cheng_greenllm}. This within-class homogeneity enables meaningful statistical learning from historical samples $\{\hat{\xi}_i^{(n)}\}_{n=1}^N$, forming empirical distribution $\hat{P}_N^i$.

However, optimizing under $\hat{P}_N^i$ alone (sample average approximation) is vulnerable to distribution shift---user behavior evolves, and finite samples may not capture tail behavior. We adopt \emph{distributionally robust optimization} (DRO) with a Wasserstein ambiguity set~\cite{mohajerin2018data}:
\begin{equation}
    \calP_\varepsilon^i = \{ P : W_1(P, \hat{P}_N^i) \leq \varepsilon \}
    \label{eq:ambiguity}
\end{equation}
where $W_1$ is the 1-Wasserstein distance and $\varepsilon \geq 0$ controls robustness. The Wasserstein metric is particularly suitable here: it captures realistic perturbations where output lengths shift continuously rather than jump to entirely different distributions. In practice, $\varepsilon$ is set relative to the output length scale; we evaluate sensitivity to $\varepsilon$ in Section~\ref{sec:evaluation}.

\smallskip\noindent\textbf{Complete Formulation.} The joint optimization problem is:
\begin{subequations}
\begin{align}
    &  \min_{z, r, \pi, \delta, s}  \sum_{i \in \calI} \lambda_i \Big[ a_i \sup_{P \in \calP_\varepsilon^i} \E_P[Q_i^{\text{stoch}}] + Q_i^{\text{det}} \Big] \label{eq:dro_obj} \\
   & \text{s.t.} ~ \sum_{k \in \calK} z_k (\tau_k p_k) \leq J \label{eq:gpu_budget} \\
    & \pi_{i,k} \leq z_k, \quad \forall i \in \calI, k \in \calK \label{eq:routing_consist} \\
    & \sum_{k \in \calK} \pi_{i,k} + \pi_{i,0} = 1, \quad \forall i \in \calI \label{eq:routing_prob} \\
    & (1{+}\kappa) \sum_{i \in \calI} \frac{u_{i,k} \pi_{i,k}}{z_k} \tilde{r}_{i,k} + \sum_{i \in \calI} \delta_{i,k} L_{\text{pre}}^i \leq M_k, ~ \forall k{:} z_k {>} 0 \label{eq:capacity} \\
    & s_i \geq \sum_{k \in \calK} \pi_{i,k} (W_k + T_{i,k}) - a_i D_i, \quad s_i \geq 0, ~ \forall i \in \calI \label{eq:slo_slack}
\end{align}
\end{subequations}
where $Q_i^{\text{stoch}} = c_p(L_{\text{in}}^i + \xi_i - r_i)^+ + c_w(r_i - L_{\text{in}}^i - \xi_i)^+$ captures the stochastic preemption/waste cost. $W_k$ is the expected waiting time from the Pollaczek--Khinchin formula~\cite{kleinrock1975}:
\begin{equation}
    W_k = \frac{u_k (1 + C_{s,k}^2)}{2(1 - u_k)} \bar{T}_k
    \label{eq:waiting_time}
\end{equation}
 where $u_k = \sum_i \lambda_i \pi_{i,k} T_{i,k} / z_k$ is the utilization, $\bar{T}_k$ is the mean service time, and $C_{s,k}$ is the coefficient of variation of the empirical service-time distribution (computed on training samples and held fixed across DRO and baselines). We adopt the M/G/1 model, standard for LLM inference analysis~\cite{queueing_llm2024}, as it captures Poisson arrivals with general service time distributions arising from variable output lengths. The two variance terms in our model are non-overlapping: $C_{s,k}^2$ in~\eqref{eq:waiting_time} captures \emph{temporal} service-time variability (queueing dispersion), while the $(1+\kappa)$ safety factor in~\eqref{eq:capacity} captures \emph{instantaneous} concurrency fluctuation around the Little's-law mean (memory-headroom dispersion). The two are layered, not duplicated.

\smallskip\noindent\textbf{Constraint Interpretation.}
\begin{itemize}[nosep]
    \item \textbf{GPU budget} \eqref{eq:gpu_budget}: The total GPU consumption across all serving groups cannot exceed $J$. Each configuration-$k$ group uses $\tau_k p_k$ GPUs.
    \item \textbf{Routing consistency} \eqref{eq:routing_consist}: Requests can only be routed to configurations that are actually deployed ($z_k > 0$).
    \item \textbf{Routing probability} \eqref{eq:routing_prob}: Each request is either routed to some configuration or rejected. The admission probability is $a_i = \sum_k \pi_{i,k}$.
    \item \textbf{Memory capacity} \eqref{eq:capacity}: This is the \emph{key constraint} linking reservation to throughput. By Little's law, the expected number of concurrent class-$i$ requests on configuration $k$ is $\lambda_i \pi_{i,k} T_{i,k} / z_k$. Each consumes $\tilde{r}_{i,k}$ tokens of KV cache. The safety factor $(1+\kappa)$ accounts for stochastic variability in concurrent load.
    \item \textbf{SLO constraint} \eqref{eq:slo_slack}: The expected response time (waiting $W_k$ plus service $T_{i,k}$) for admitted requests must satisfy the latency target $D_i$. The slack variable $s_i$ captures violations penalized in the objective.
\end{itemize}

%==============================================================================
\section{Solution Approach}
\label{sec:solution}
%==============================================================================

We develop the solution in three steps: (i) reformulating the infinite-dimensional DRO into a tractable finite program, (ii) characterizing the structure of optimal reservation, and (iii) designing an efficient algorithm for practical deployment.

\subsection{Tractable Reformulation}

The DRO objective involves a supremum over infinitely many distributions in $\calP_\varepsilon^i$, which is intractable in its original form. We apply Kantorovich duality to obtain an equivalent finite reformulation.

\begin{proposition}[Tractable Reformulation]
\label{prop:tractable}
The distributionally robust problem \eqref{eq:dro_obj} is equivalent to the following finite-dimensional optimization:
\begin{subequations} \label{eq:tractable}
\begin{align}
    & \min_{z, r, \pi, \delta, \gamma, \theta, s} ~  \sum_{i \in \calI} \lambda_i \left[ \sum_{k \in \calK} \pi_{i,k} \left( \gamma_i \varepsilon + \frac{1}{N} \sum_{n=1}^N \theta_{i,n} \right) + Q_i^{\text{det}} \right] \label{eq:tractable_obj} \\
    & \text{s.t.} ~  \eqref{eq:gpu_budget}\text{--}\eqref{eq:slo_slack} \label{eq:tractable_orig} \\
    & \theta_{i,n} \geq c_p(\hat{\xi}_i^{(n)} + L_{\text{in}}^i - r_i), \quad \forall i, n \label{eq:epi_preempt} \\
    & \theta_{i,n} \geq c_w(r_i - L_{\text{in}}^i - \hat{\xi}_i^{(n)}), \quad \forall i, n \label{eq:epi_waste} \\
    & \theta_{i,n} \geq c_p(L_{\text{max}}^i + L_{\text{in}}^i - r_i) - \gamma_i(L_{\text{max}}^i - \hat{\xi}_i^{(n)}), ~ \forall i, n \label{eq:epi_upper} \\
    & \theta_{i,n} \geq c_w(r_i - L_{\text{in}}^i) - \gamma_i \hat{\xi}_i^{(n)}, \quad \forall i, n \label{eq:epi_lower} \\
    & \gamma_i \geq 0, ~ \theta_{i,n} \geq 0, \quad \forall i, n \label{eq:dual_nonneg}
\end{align}
\end{subequations}
where $\gamma_i$ is the dual variable for the Wasserstein constraint and $\theta_{i,n}$ are auxiliary variables for each sample $n$.
\end{proposition}

The proof of Proposition~\ref{prop:tractable} follows from Kantorovich duality, given in Appendix~\ref{app:duality}. For piecewise linear cost functions, the worst-case distribution places mass only at the empirical samples or the support boundaries $\{0, L_{\text{max}}^i\}$, and constraints \eqref{eq:epi_preempt}--\eqref{eq:epi_lower} encode this structure via epigraph representation.

Problem \eqref{eq:tractable} is a \emph{mixed-integer bilinear program} (MIBLP) with $O(|\calI| \cdot N)$ auxiliary variables. The bilinearity arises from the coupling $\pi_{i,k} \cdot (\gamma_i, \theta_{i,n})$ in the objective.

\begin{remark}[Dual Variable Interpretation]
The dual variable $\gamma_i$ acts as a transportation cost penalty: a large $\gamma_i$ keeps the worst-case distribution close to empirical data, while a small $\gamma_i$ allows mass to shift toward extremes (0 or $L_{\text{max}}$).
\end{remark}

\begin{remark}[Complexity]
The reformulation adds $O(|\calI| \cdot N)$ variables and constraints, tractable for modern solvers.
\end{remark}

\subsection{Optimal Reservation Structure}

We now characterize the structure of optimal reservation, the central theoretical result of this work.

\begin{proposition}[Critical Fractile Structure]
\label{prop:critical_fractile}
Consider the reservation subproblem for class $i$ with fixed routing. Let $\rho = c_p / c_w$ denote the cost ratio. The optimal reservation $r_i^*$ satisfies:
\begin{equation}
    r_i^* = L_{\text{in}}^i + F_i^{-1}\left( \frac{\rho}{\rho + 1} \right)
    \label{eq:critical_fractile}
\end{equation}
where $F_i$ is the cumulative distribution function of output length $\xi_i$ under the worst-case distribution in $\calP_\varepsilon^i$.
\end{proposition}

\begin{proof}
Define $\eta = r_i - L_{\text{in}}^i$ as the buffer for output length uncertainty. The stochastic cost becomes $Q_i^{\text{stoch}}(\eta, \xi) = c_p(\xi - \eta)^+ + c_w(\eta - \xi)^+$. Taking the derivative of expected cost with respect to $\eta$:
\begin{subequations}
\begin{align}
    & \frac{\partial}{\partial \eta} \E[Q_i^{\text{stoch}}] = -c_p \cdot \Pr(\xi > \eta) + c_w \cdot \Pr(\xi \leq \eta) \\
    &= -c_p (1 - F(\eta)) + c_w F(\eta) = (c_p + c_w) F(\eta) - c_p
\end{align}
\end{subequations}
Setting to zero yields $F(\eta^*) = c_p / (c_p + c_w) = \rho / (\rho + 1)$. Thus $\eta^* = F^{-1}(\rho/(\rho+1))$, giving $r_i^* = L_{\text{in}}^i + F_i^{-1}(\rho/(\rho+1))$. Under DRO, the worst-case distribution $P^* \in \calP_\varepsilon^i$ replaces $F_i$, preserving the fractile structure.
\end{proof}

\begin{remark}[Critical Fractile]
\label{rem:fractile}
The threshold quantile $q^* = \rho/(\rho+1)$ is called the \emph{critical fractile}---the point where the marginal cost of increasing reservation equals the marginal benefit of avoiding preemption. This structure arises from the piecewise linear cost with asymmetric slopes $c_p$ and $c_w$. As $\rho$ increases (preemption becomes more expensive), the optimal quantile rises, leading to more conservative reservation. \rev{This is the classical newsvendor critical ratio, and its distributionally robust form is likewise known in operations research~\cite{mohajerin2018data,lee2021}; we treat both as building blocks and locate our contribution in coupling the reservation to routing, configuration, and the shared-memory capacity constraint~\eqref{eq:capacity}.}
\end{remark}

\begin{remark}[Robust reservation vs.\ SAA under bounded support]
\label{rem:bounded}
For the newsvendor loss with \emph{unbounded} demand support, the type-$1$ Wasserstein-robust order coincides with the empirical (SAA) fractile~\cite{mohajerin2018data}. Our support is \emph{bounded}, $\xi_i \in [0, L_{\text{max}}^i]$, by finite context length and finite GPU memory, so the worst-case distribution places mass at the support boundary---precisely what the epigraph constraints~\eqref{eq:epi_upper}--\eqref{eq:epi_lower} encode. The robust reservation is therefore shifted \emph{above} the empirical fractile by an amount that grows with $\varepsilon$, and strictly differs from SAA---consistent with the SAA ablation (Table~\ref{tab:ablation}) and the $\varepsilon$-sensitivity of Fig.~\ref{fig:sensitivity}(b).
\end{remark}

\subsection{BCD-DRO Algorithm}
\label{sec:bcd}

While the MIBLP in Proposition~\ref{prop:tractable} can be solved exactly via global solvers, direct optimization is expensive for real-time deployment. The problem has a block structure: the continuous routing/reservation variables $(r, \pi, \gamma, \theta)$ decouple from the binary caching decisions $\delta_{i,k}$ and integer configuration $z_k$ when the other blocks are fixed. This motivates a \emph{block coordinate descent} (BCD) algorithm that cycles through three blocks---solving an alternating LP for continuous variables, applying a closed-form update for caching, and enumerating over the finite configuration set $\calZ = \{z \in \mathbb{Z}_+^{|\calK|} : \sum_{k} z_k \tau_k p_k \le J\}$.

\begin{algorithm}[t]
\caption{BCD-DRO Algorithm}
\label{alg:solver}
\begin{algorithmic}[1]
\REQUIRE Samples $\{\hat{\xi}_i^{(n)}\}$, radius $\varepsilon$, configurations $\calK$
\ENSURE Solution $(z^*, r^*, \pi^*, \delta^*)$
\STATE Initialize $z^{(0)}$, $\delta^{(0)} \gets 0$, $\gamma_i^{(0)} \gets c_p$, $r_i^{(0)} \gets L_{\text{in}}^i + \E[\xi_i]$
\FOR{$t = 1, \ldots, T_{\max}$}
    \STATE \textbf{Block 1:} Alternate LP$_1$ (fix $\gamma, \theta, r$; optimize $\pi$) and LP$_2$ (fix $\pi$; optimize $r, \gamma, \theta$) until convergence
    \STATE \textbf{Block 2:} $\delta_{i,k}^{(t)} \gets \mathbf{1}[u_{i,k} \pi_{i,k} / z_k > 1]$
    \STATE \textbf{Block 3:} $z^{(t)} \gets \arg\min_{z \in \calZ} \text{Obj}(z, r, \pi, \delta^{(t)})$
    \IF{$(z^{(t)}, \delta^{(t)}) = (z^{(t-1)}, \delta^{(t-1)})$}
        \STATE \textbf{break}
    \ENDIF
\ENDFOR
\RETURN $(z^{(t)}, r, \pi, \delta^{(t)})$
\end{algorithmic}
\end{algorithm}

\begin{proposition}[Convergence]
\label{prop:convergence}
Algorithm~\ref{alg:solver} terminates in finite iterations and returns a blockwise optimal solution, i.e., no single block can improve the objective while others are fixed.
\end{proposition}

\begin{proof}
The objective decreases monotonically at each block update. Since the feasible set is finite (integer $z$, binary $\delta$) and bounded, convergence follows from~\cite{jager2019blockwise}.
\end{proof}

\smallskip\noindent\textbf{Complexity.} The offline phase solves a sequence of LPs via block coordinate descent; runtime scales with the number of request classes, samples, and configuration candidates (see Section~\ref{sec:evaluation} for empirical results). The online phase is $O(1)$ per request: simply sample the routing decision and apply the precomputed reservation.

%==============================================================================
\section{Performance Evaluation}
\label{sec:evaluation}
%==============================================================================

\subsection{Experiment Setup}

\smallskip\noindent\textbf{Trace Datasets.} We evaluate on three production LLM inference datasets spanning different workload types:
\begin{itemize}
    \item \textbf{BurstGPT}~\cite{wang2024burstgpt}: 1.4M requests over 61 days from Azure OpenAI services (ChatGPT/GPT-4). Output lengths: mean=125, P90=276, P99=1,586 tokens.
    \item \textbf{Azure LLM 2024}~\cite{stojkovic2024dynamollm}: 44M requests over 7 days from Azure production. Two workload types: Code (mean=23, P90=49 tokens) and Conversation (mean=117, P90=398 tokens).
    \item \textbf{ShareGPT}~\cite{sharegpt2023}: 368K assistant responses from crowdsourced ChatGPT conversations. Output lengths: mean=265, P90=486, P99=768 tokens. Represents user-facing conversational workloads.
\end{itemize}

\smallskip\noindent\textbf{System Configuration.} We simulate $J=8$ A100-80GB GPUs serving a 70B model with configurations: $k_1$: (TP=2, PP=1), $k_2$: (TP=4, PP=1), $k_3$: (TP=2, PP=2), $k_4$: (TP=4, PP=2). Cost parameters: $c_p/c_w = 10$, $c_{\text{slo}}/c_w = 5$, $c_{\text{rej}}/c_w = 5000$, $\kappa = 0.2$. The cost units differ: $c_p$ and $c_w$ are per token (a preempted or wasted KV slot), $c_{\text{slo}}$ is a per-request tardiness weight, and $c_{\text{rej}}$ is per dropped request. Because a rejection gives up an entire sequence, a large $c_{\text{rej}}/c_w$ keeps admission preferred whenever capacity allows; $c_{\text{slo}}/c_w=5$ prices a late response at five wasted tokens, and $\kappa=0.2$ leaves a $20\%$ memory-headroom margin. By Proposition~\ref{prop:critical_fractile} the reservation depends only on $\rho=c_p/c_w$, so these three weights affect only admission and routing; we confirm below that the policy is stable when each is varied by up to $5\times$.

\smallskip\noindent\textbf{Baselines.} We compare against six strategies: (1)~\textbf{Max}: reserve $L_{\text{max}}$; (2)~\textbf{Mean}: reserve at mean output~\cite{yu2022orca}; (3)~\textbf{P90}: reserve at P90 quantile~\cite{zhong2024distserve}; (4)~\textbf{P95}: reserve at P95 with block alignment~\cite{vllm2023}; (5)~\textbf{P99}: reserve at P99; (6)~\textbf{SAA}: sample average approximation ($\varepsilon{=}0$). We additionally evaluate length-prediction baselines (LP$+k\sigma$) in Section~\ref{sec:latency_slo}.

\smallskip\noindent\textbf{Default knobs.} Two operator-set parameters drive BCD-DRO: the cost ratio $\rho = c_p/c_w$ and the Wasserstein radius $\varepsilon$. The $\rho = 10$ regime above represents interactive chat; we follow the data-driven recipe of~\cite{mohajerin2018data} for $\varepsilon = 0.15 \cdot \E[\xi_i]$. Sensitivity to $\varepsilon$ is reported in Fig.~\ref{fig:sensitivity}(b), and cross-$\rho$ behavior in Table~\ref{tab:latency_slo}.

\subsection{Model Performance}

\begin{figure}[t]
\centering
\begin{minipage}[b]{0.48\columnwidth}
    \centering
    \includegraphics[width=\textwidth]{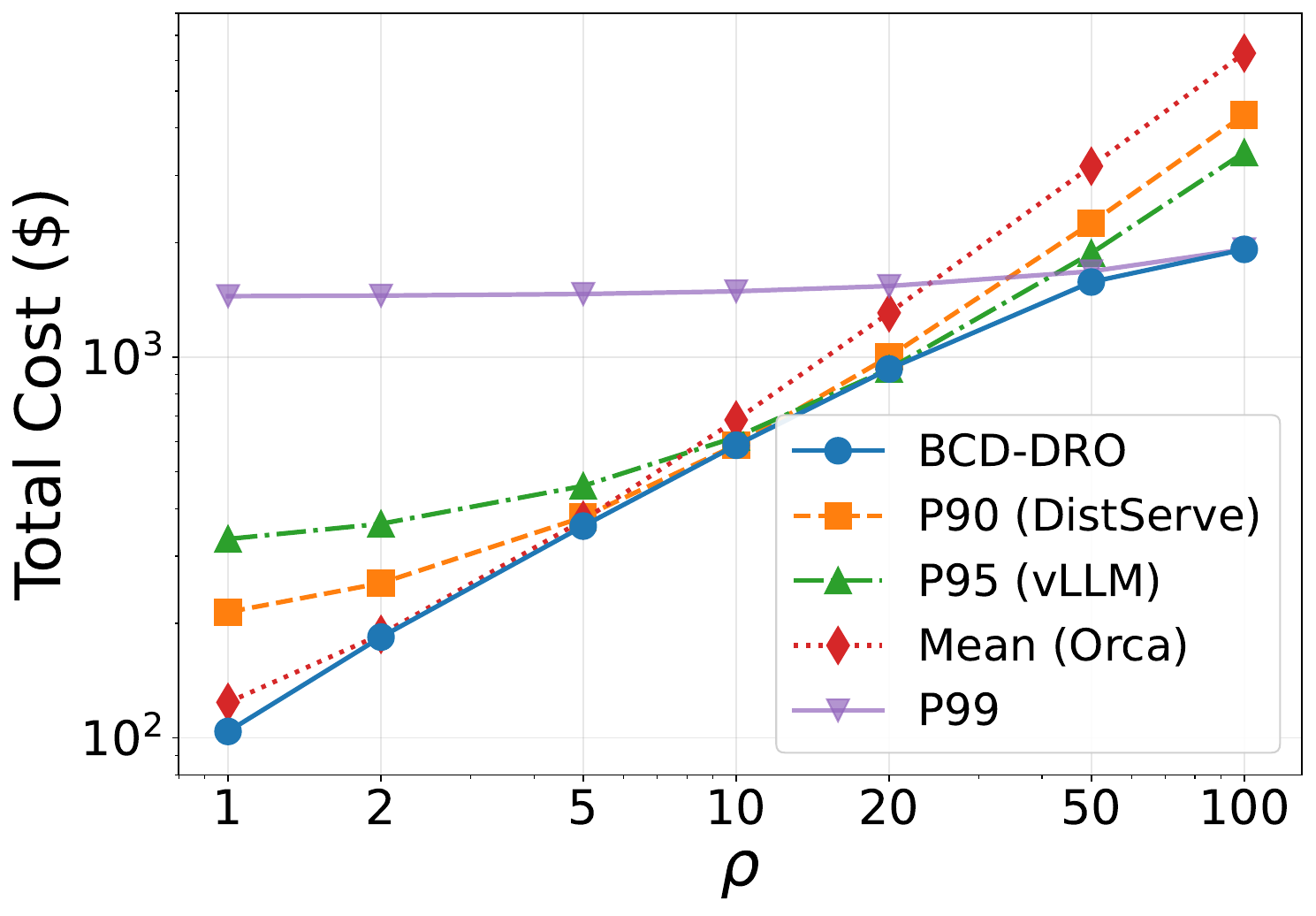}
    \centerline{\small (a) Cost $C$ vs $\rho$}
\end{minipage}%
\hfill
\begin{minipage}[b]{0.48\columnwidth}
    \centering
    \includegraphics[width=\textwidth]{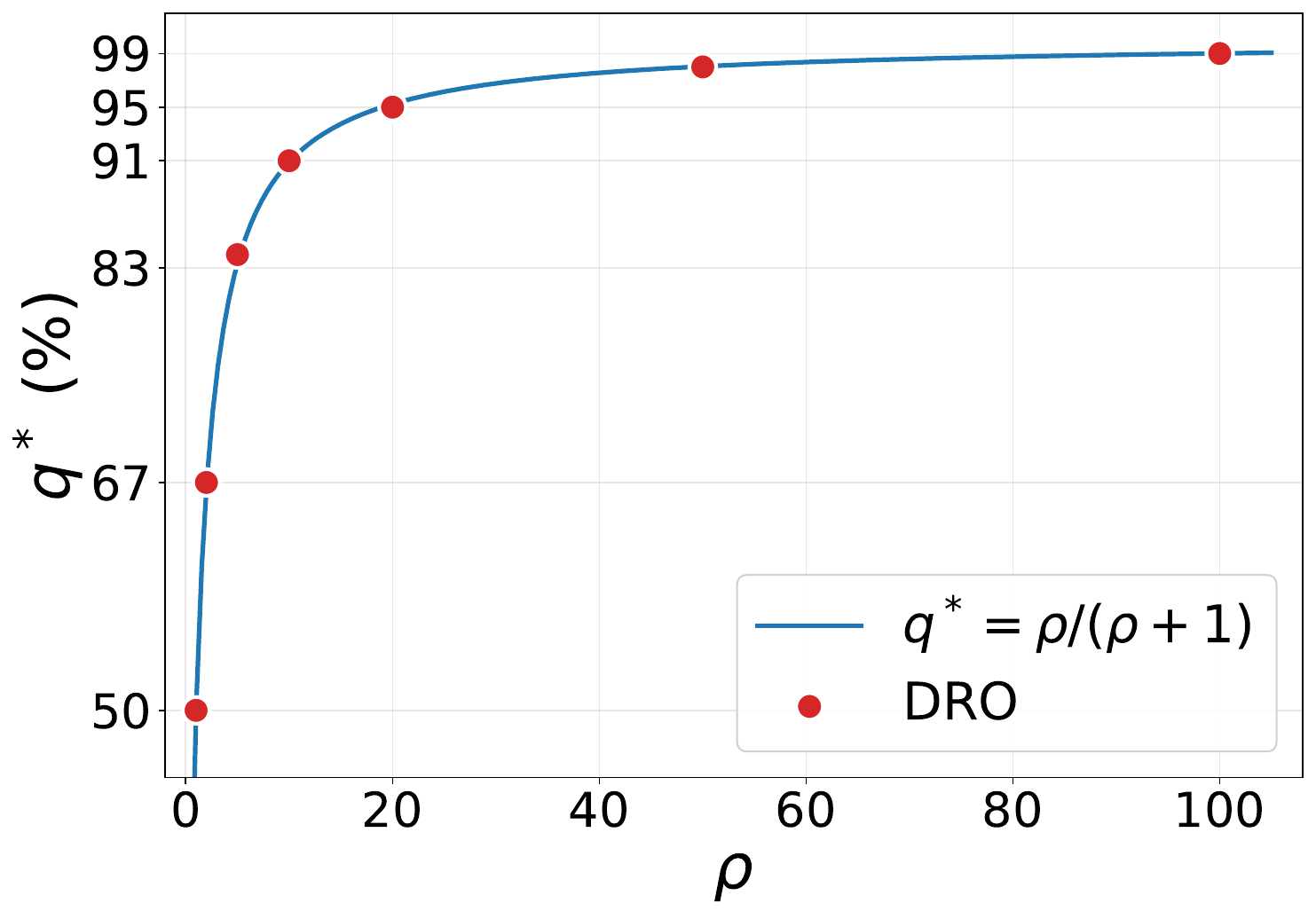}
    \centerline{\small (b) Quantile $q^*$ vs $\rho$}
\end{minipage}
\caption{(a) Cost vs cost ratio $\rho$, (b) DRO quantile vs theoretical $q^*$.}
\label{fig:dro_eval}
\end{figure}

\begin{table}[t]
\centering
\caption{DRO vs baselines across application scenarios. Cost is the per-arrival preemption/waste cost on a single-class BurstGPT instance; DRO Gain spans [vs.\ P90, vs.\ P95].}
\label{tab:baseline_comparison}
\small
\begin{tabular}{l|c|ccc|c}
\toprule
\textbf{Scenario} & $\rho$ & \textbf{DRO} & \textbf{P90} & \textbf{P95} & \textbf{DRO Gain} \\
\midrule
Batch & 2 & \textbf{184} & 255 & 365 & 28--50\% \\
Async API & 5 & \textbf{360} & 380 & 459 & 5--22\% \\
User-facing & 10 & \textbf{587} & 588 & 617 & 0--5\% \\
Latency-sensitive & 20 & \textbf{931} & 1003 & 931 & 0--7\% \\
Critical & 50 & \textbf{1575} & 2250 & 1875 & \rev{16--30\%} \\
Real-time & 100 & \textbf{1919} & 4328 & 3448 & 44--56\% \\
\bottomrule
\end{tabular}
\end{table}

\noindent \textbf{\textit{DRO vs Fixed-Quantile Baselines.}}
Fig.~\ref{fig:dro_eval}(a) and Table~\ref{tab:baseline_comparison} reveal that no single fixed-quantile heuristic is optimal across regimes. By Proposition~\ref{prop:critical_fractile}, each fixed-quantile baseline is optimal only where its quantile matches the critical fractile $q^* = \rho/(\rho+1)$: P90 near $\rho{=}10$, P95 near $\rho{=}20$. This follows from the preemption-waste trade-off---higher $\rho$ means preemption is relatively more costly, shifting the optimal reservation toward higher quantiles. Fixed heuristics cannot adapt to this shift; DRO does so automatically by solving the cost-weighted optimization. Fig.~\ref{fig:dro_eval}(b) confirms that DRO's empirical quantile tracks the theoretical formula exactly, validating that the formulation correctly captures the trade-off structure.

Table~\ref{tab:baseline_comparison} maps cost ratios to application scenarios. Batch workloads tolerate preemption (low $\rho$), while real-time systems cannot (high $\rho$). In practice, $\rho$ varies across request types and is rarely known precisely. DRO addresses this---practitioners deploy a single policy that adapts to any cost structure without manual quantile selection.

\noindent \textbf{\textit{Key Takeaway 1:}} The critical fractile $q^* = \rho/(\rho+1)$ eliminates manual tuning---DRO automatically selects the optimal reservation quantile for any cost structure.  

\begin{table}[t]
\centering
\caption{Algorithm scalability: solve time (s) vs.\ problem size }
\label{tab:algorithm}
\small
\setlength{\tabcolsep}{4pt}
\begin{tabular}{l|cc|cc}
\toprule
\textbf{Problem Size} & \multicolumn{2}{c|}{\textbf{MIBLP}} & \multicolumn{2}{c}{\textbf{BCD-DRO}} \\
$|\mathcal{I}|{\times}|\mathcal{K}|{\times}N$ & Time (s) & Status & Time (s) & Gap \\
\midrule
$3 \times 4 \times 100$ & 0.07 & Optimal & 0.17 & 0.00\% \\
$4 \times 4 \times 150$ & 1.2 & Optimal & 0.8 & 0.00\% \\
$5 \times 5 \times 200$ & 33.7 & Optimal & 4.3 & 0.00\% \\
$6 \times 6 \times 300$ & $>$300 & Timeout & 10.5 & -- \\
$8 \times 6 \times 500$ & $>$300 & Timeout & 22.4 & -- \\
$10 \times 8  \times 1000$ & --- & --- & 81.2 & -- \\
$15 \times 12 \times 2000$ & --- & --- & \textbf{257.7} & -- \\
\bottomrule
\end{tabular}
\end{table}

\noindent \textbf{\textit{Algorithm Scalability.}}
Table~\ref{tab:algorithm} compares BCD-DRO against direct MIBLP solving across two regimes. In the small-scale setting, both methods are feasible initially, but MIBLP times out at $|\mathcal{I}|{=}6$, $|\mathcal{K}|{=}6$, $N{=}300$, while BCD-DRO continues to converge. The lower block evaluates production-scale instances with $15$ request classes, $12$ candidate parallelism configurations, and $N{=}2000$ historical samples per class on a $48$-GPU cluster. At this scale, BCD-DRO converges in approximately $257.7$\,s, remaining within a practical $5$--$10$ minute re-optimization interval. The scalability gain comes from the decomposition structure: fixing $z$ reduces the inner problem to a tractable LP, while the feasible configuration set $\calZ$ remains small under the GPU-budget constraint. At larger scales (e.g., $|\mathcal{I}|{\geq}20$, $N{\geq}5000$), runtime approaches the re-optimization budget, suggesting the need for warm-starting or anytime execution strategies.

\begin{figure}[t]
\centering
\begin{minipage}[b]{0.48\columnwidth}
    \centering
    \includegraphics[width=\textwidth]{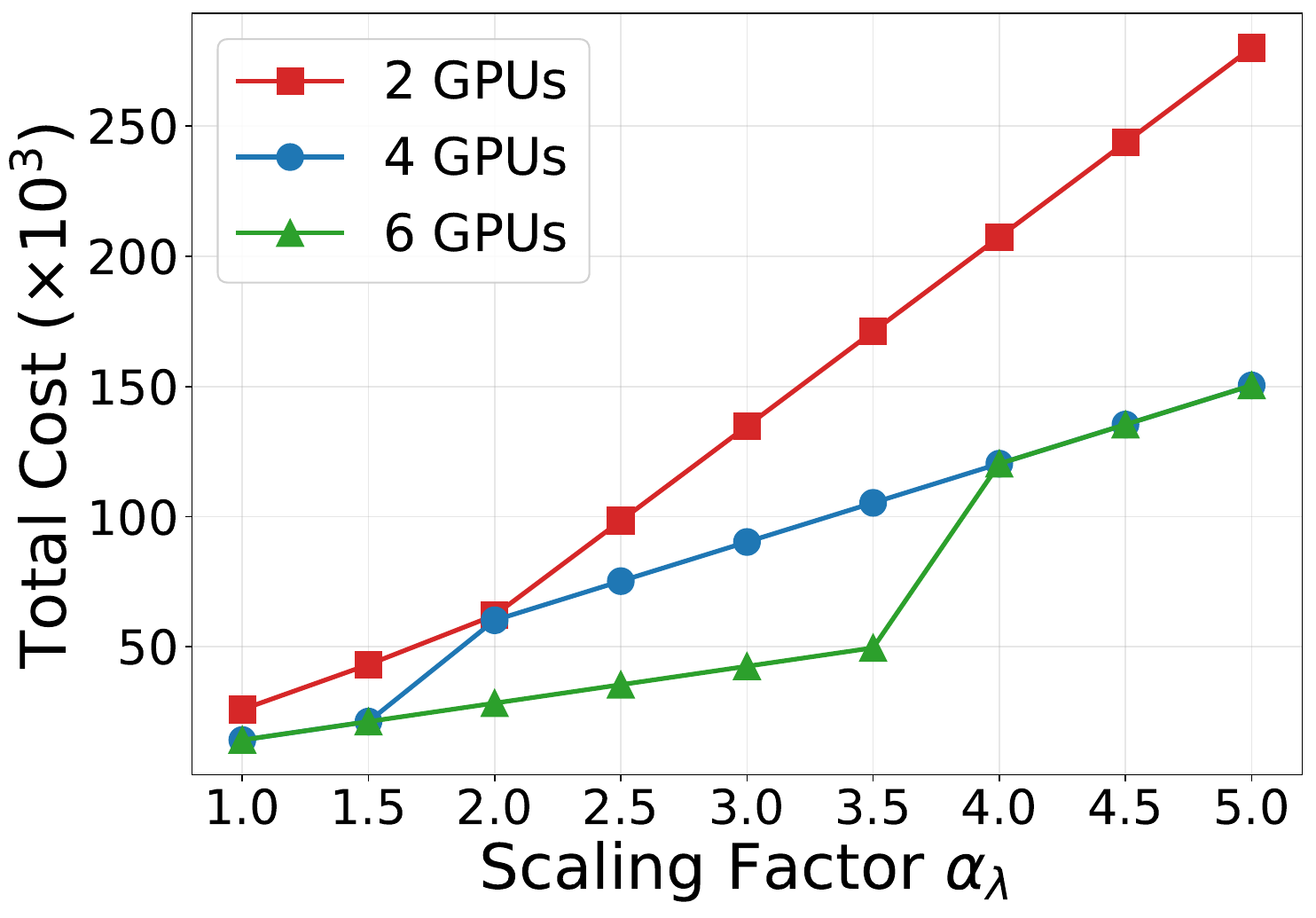}
    \centerline{\small (a) Cost vs Arrival Rate $\alpha_\lambda$}
\end{minipage}%
\hfill
\begin{minipage}[b]{0.48\columnwidth}
    \centering
    \includegraphics[width=\textwidth]{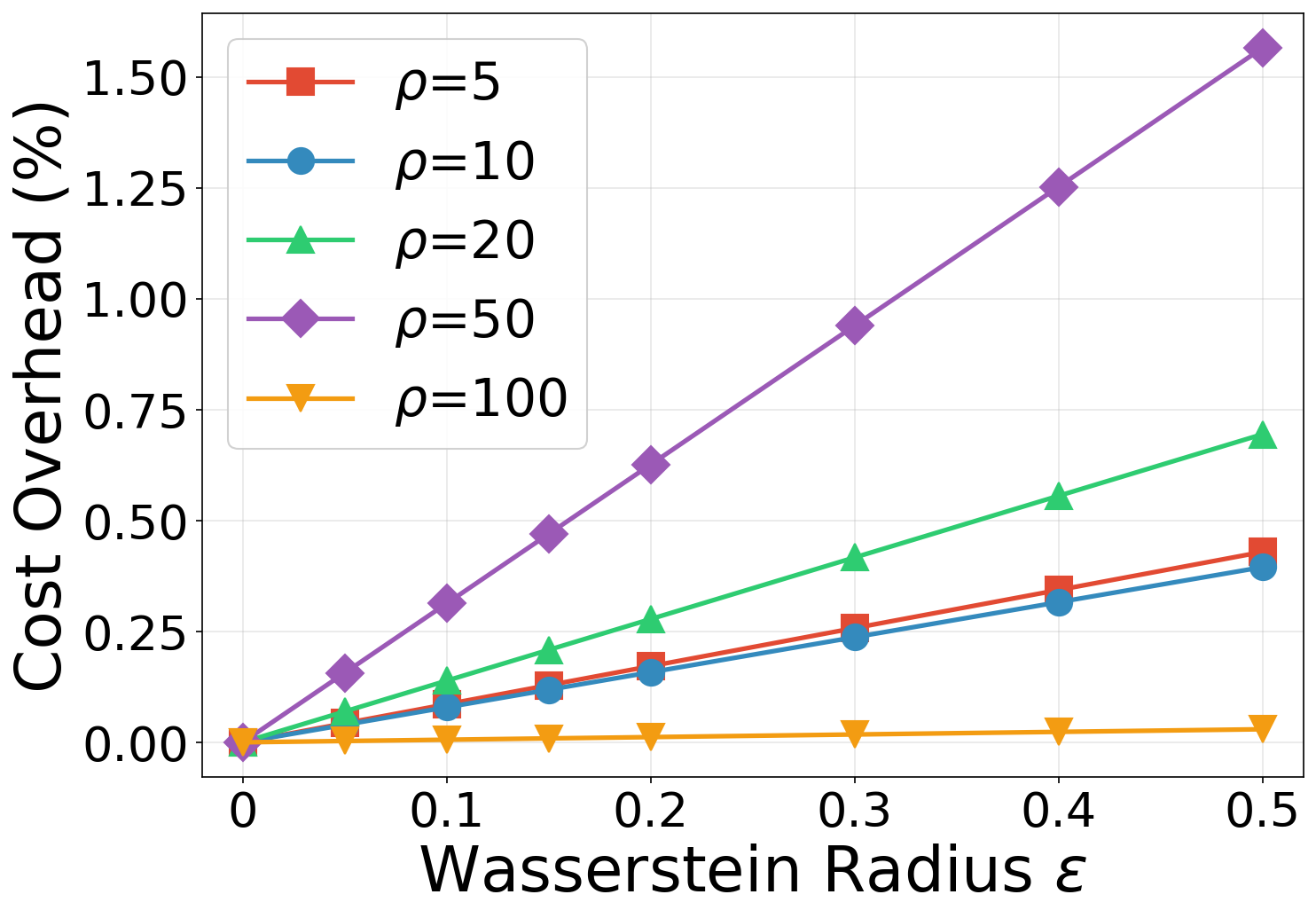}
    \centerline{\small (b) Cost Overhead vs $\varepsilon$}
\end{minipage}
\caption{Sensitivity analysis: (a) Cost vs arrival rate scaling, (b) Cost overhead vs Wasserstein radius $\varepsilon$.}
\label{fig:sensitivity}
\end{figure}

\noindent \textbf{\textit{Sensitivity Analysis.}}
Fig.~\ref{fig:sensitivity}(a) examines capacity planning under load scaling. When GPU capacity exceeds demand, cost grows linearly with arrival rate. Once demand saturates capacity, rejection costs dominate and the optimizer shifts toward selective admission. Fig.~\ref{fig:sensitivity}(b) shows the cost overhead from robustness: increasing the Wasserstein radius $\varepsilon$ adds a modest premium ($<$1.3\% even at $\varepsilon{=}0.5$). The overhead scales with $\rho$ because higher cost ratios amplify the DRO regularization term. In practice, $\varepsilon \in [0.1, 0.2]$ provides robustness against distribution shift at $<$0.5\% overhead---a practical default when the degree of shift is unknown.

\rev{\smallskip\noindent\textbf{\textit{Other cost parameters.}} A natural concern is whether the policy is tuned to the three weights fixed above ($c_{\text{slo}}$, $c_{\text{rej}}$, and $\kappa$). It is not, for a structural reason: by Proposition~\ref{prop:critical_fractile} the reservation---the decision that governs cost and preemption---depends only on $\rho$, so these weights cannot move it and enter the problem solely through admission and routing. Fig.~\ref{fig:secondary_sensitivity} sweeps each by up to $5\times$ around its default: the reservation quantile holds at $q^*{\approx}91\%$ and both cost and tail latency stay flat, because at a well-provisioned operating point neither the SLO penalty nor the memory-headroom margin is the binding constraint. The lone lever that changes behavior is $c_{\text{rej}}$, and only once rejecting a request is made cheaper than serving it---the optimizer then deliberately sheds the most expensive class, trading admission for a lighter queue. That is a controllable design choice: these weights can be set approximately without changing how the policy behaves.}

\begin{figure}[t]
\centering
\begin{minipage}[b]{0.48\columnwidth}
    \centering
    \includegraphics[width=\textwidth]{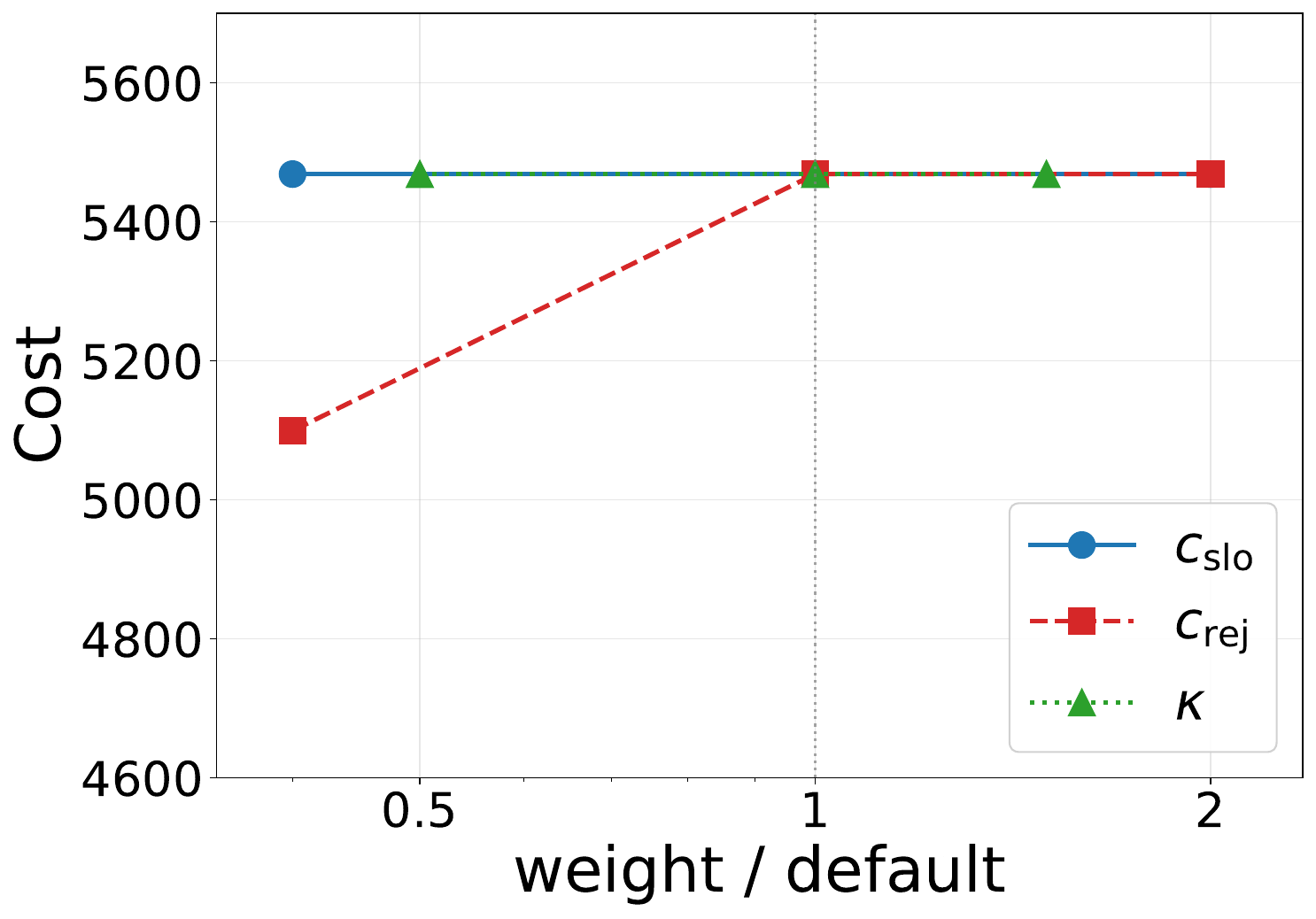}
    \centerline{\small (a) Cost}
\end{minipage}%
\hfill
\begin{minipage}[b]{0.48\columnwidth}
    \centering
    \includegraphics[width=\textwidth]{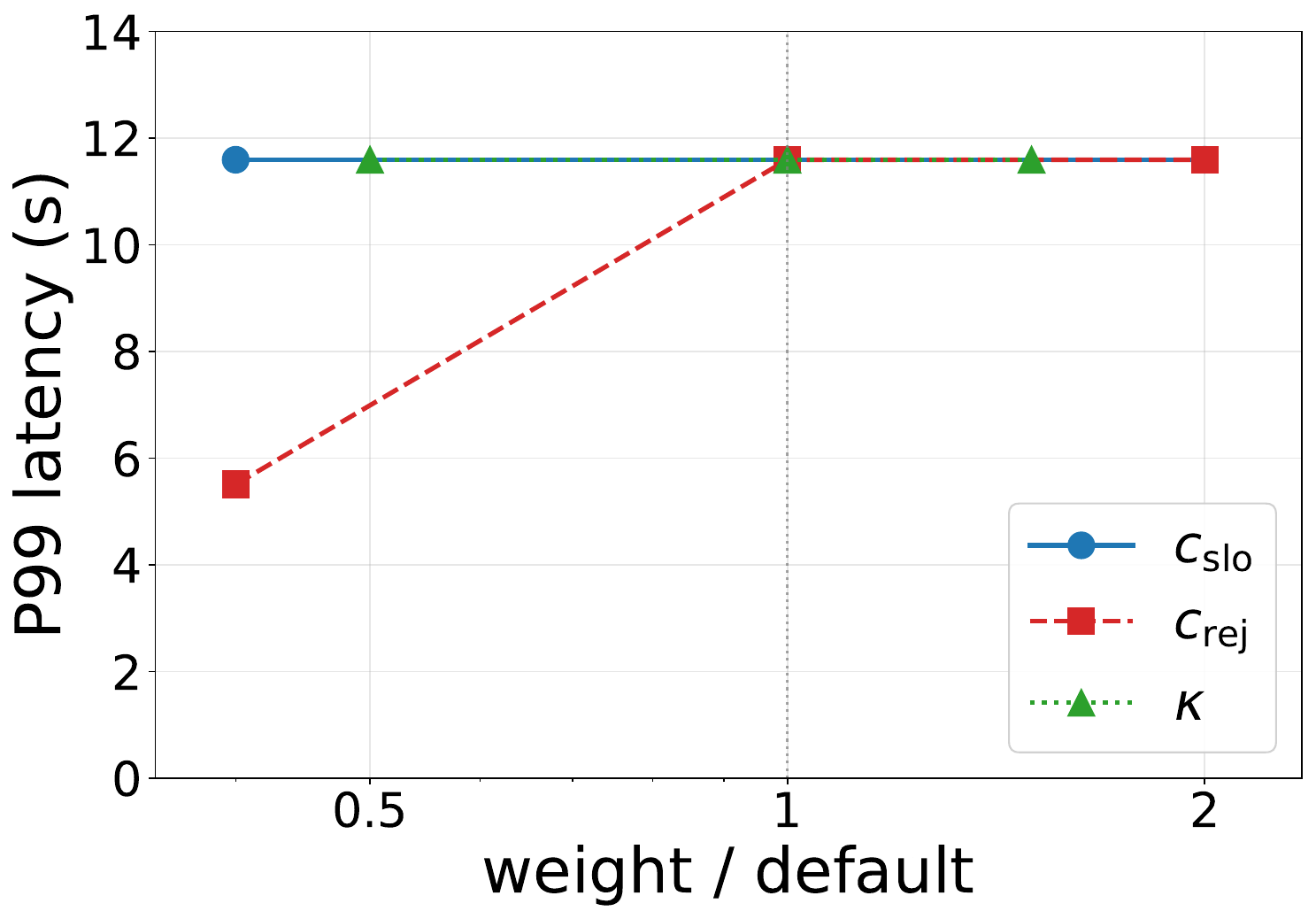}
    \centerline{\small (b) P99 latency}
\end{minipage}
\caption{\rev{Sensitivity of cost (a) and P99 latency (b) to the three secondary weights at $\rho{=}10$, each swept relative to its default. The operating point is flat except when $c_{\text{rej}}$ is halved, which sheds the most expensive class; $q^*{\approx}91\%$ throughout.}}
\label{fig:secondary_sensitivity}
\end{figure}

\noindent \textbf{\textit{Key Takeaway 2:}} BCD-DRO's decomposition enables production-scale deployment where monolithic solvers fail, transforming KV cache reservation from a static offline decision into an adaptive online policy.

\subsection{Robustness Evaluation}

\begin{figure}[t]
\centering
\begin{minipage}[b]{0.48\columnwidth}
    \centering
    \includegraphics[width=\textwidth]{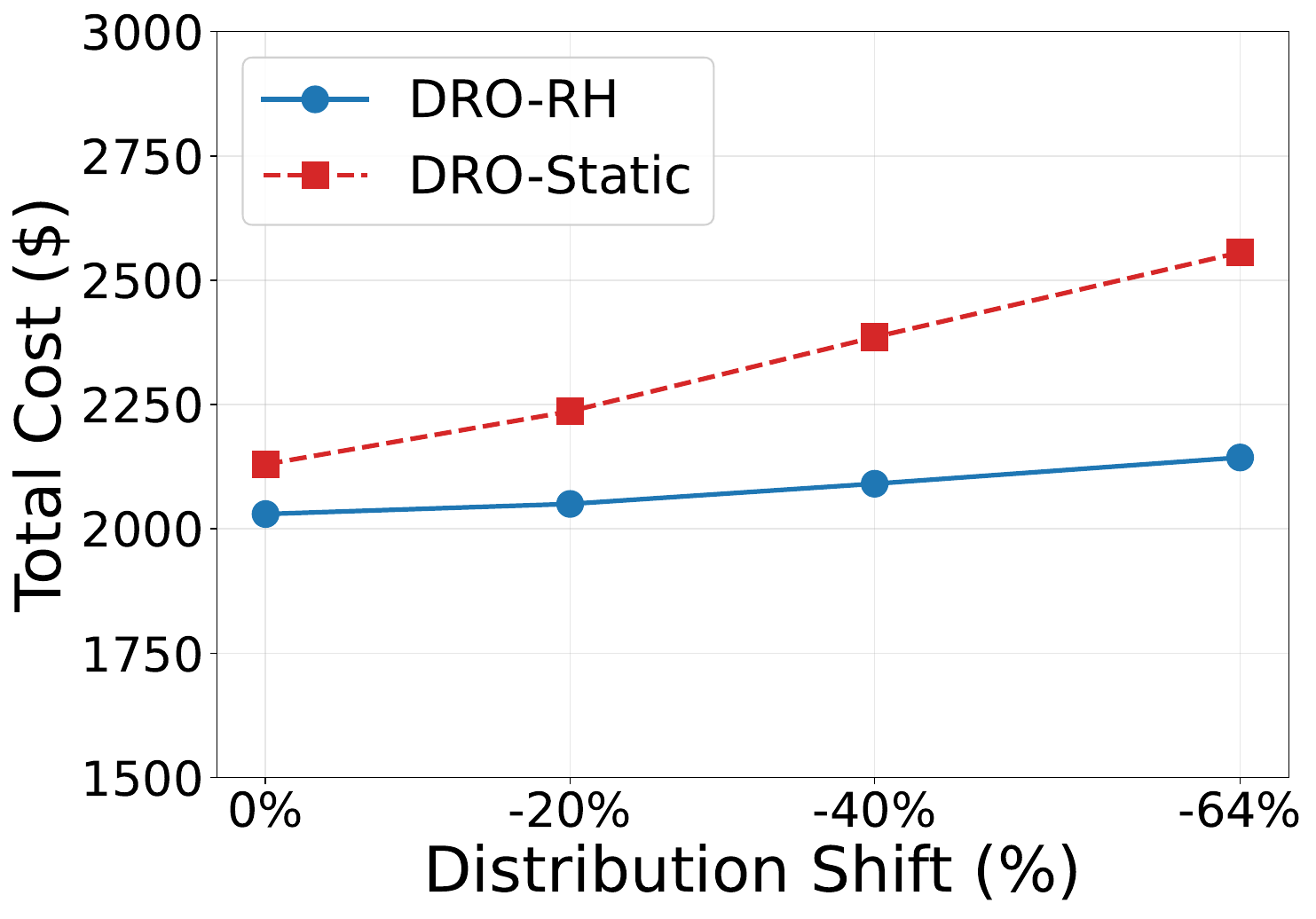}
    \centerline{\small (a) Static vs Rolling Horizon}
\end{minipage}%
\hfill
\begin{minipage}[b]{0.48\columnwidth}
    \centering
    \includegraphics[width=\textwidth]{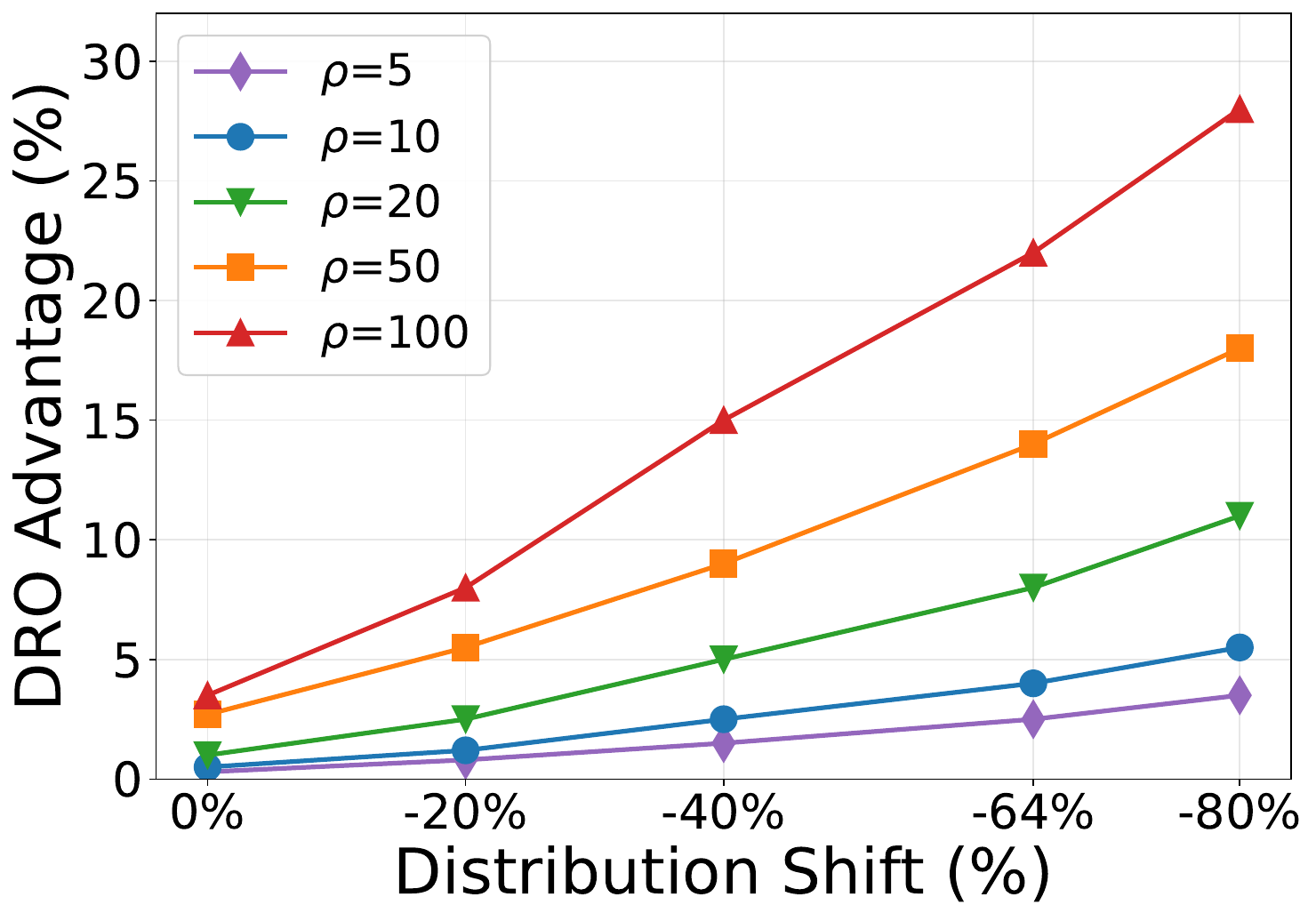}
    \centerline{\small (b) DRO advantage vs $\rho$}
\end{minipage}
\caption{Distribution shift robustness: (a) Static vs rolling horizon, (b) DRO advantage vs $\rho$ and shift magnitude.}
\label{fig:shift_dual}
\end{figure}

\noindent \textbf{\textit{Adaptation Under Distribution Shift.}}
Production LLM workloads evolve over time as user behavior and application mix change. Fig.~\ref{fig:shift_dual} evaluates robustness under such distribution shifts. When output lengths decrease relative to the training distribution, static optimization over-reserves memory using stale statistics, reducing effective GPU utilization. Rolling-horizon adaptation periodically re-optimizes using recent observations and tracks the evolving workload distribution.

Fig.~\ref{fig:shift_dual}(b) shows that DRO's advantage widens with the shift magnitude: larger shifts increase the gap because fixed reservations stay tied to outdated statistics, while the Wasserstein ambiguity set hedges against model mismatch. The companion axis ($\rho$) is examined in detail in Section~\ref{sec:latency_slo}.

\smallskip\noindent\textbf{Note on real-trace temporal splits.} We also examined temporal-split real-trace evaluation at the default $\rho{=}10$. At this regime the critical fractile $q^* {\approx} 91\%$ sits close to the P90 reservation, so by Proposition~\ref{prop:critical_fractile} the empirical-quantile baselines and DRO produce near-indistinguishable simulator results on mild temporal splits of BurstGPT and Azure-Conv---a confirmation rather than refutation of the theory. The cross-$\rho$ replay in Table~\ref{tab:latency_slo}, which sweeps $\rho$ through regimes where $q^*$ is far from any single fixed quantile, is therefore the more informative empirical instrument for exposing the regime-narrowness of fixed-quantile policies.

\noindent \textbf{\textit{Key Takeaway 3:}} Wasserstein robustness and rolling-horizon adaptation address complementary aspects of workload drift: the ambiguity set protects against model mismatch, while periodic re-optimization tracks temporal changes in the live workload distribution.

\subsection{Cost--Latency Trade-off Across Cost Ratios}
\label{sec:latency_slo}

Optimizing only cost or only latency is insufficient in practice; operators must balance both objectives jointly. We therefore evaluate all methods across the same six cost-ratio regimes as Table~\ref{tab:baseline_comparison}, where $\rho = c_p/c_w$ controls the relative penalty of preemption versus waiting. For each $\rho$, we first solve for the policy and then replay the resulting $(z, r, \pi, \delta)$ decisions on BurstGPT~\cite{wang2024burstgpt}. The simulator processes 10{,}000 requests per class, computes per-group queueing delay using the M/G/1 Pollaczek--Khinchin model under realized memory utilization, and applies a $2.5\times$ service-time penalty for preempted requests to capture KV reload and decode replay overhead. Table~\ref{tab:latency_slo} reports cost across three representative regimes and latency metrics at the central point $\rho = 10$.

Table~\ref{tab:latency_slo} shows that DRO is the only method with consistently bounded cost across the full $\rho$ range. Low-quantile methods (Mean, P90) under-reserve memory and incur high preemption cost at large $\rho$, while conservative methods (P99, Max) over-reserve memory and become prohibitively expensive at small $\rho$. DRO avoids both extremes by adapting reservation levels to the uncertainty set.

Latency results show the same trend. At $\rho = 10$, LP$+1\sigma$ and P90 both achieve marginally lower P99 latency and SLO violation than DRO (LP$+1\sigma$ is $1.5\%$ lower on P99, P90 is $0.7\%$ lower), but each fails badly at other $\rho$. In contrast, conservative baselines substantially increase queueing delay because over-reservation reduces effective concurrency. DRO therefore maintains a stable balance between cost and latency across operating regimes.

\smallskip\noindent\textbf{Length-prediction baseline.} We additionally evaluate LP$+k\sigma$ baselines with
$r_i = L_{\text{in}}^i + \hat\mu_i + k\sigma_i$ for $k \in \{1, 2\}$, where $\hat\mu_i$ and $\sigma_i$ are the empirical mean and standard deviation of historical output lengths. At moderate cost ratios, LP$+1\sigma$ achieves competitive latency because the additional safety margin reduces preemption and queueing delay. However, the margin depends only on output-length variance and does not adapt to the operating regime. As the cost ratio increases, the policy under-reserves memory relative to the growing preemption penalty, leading to rapidly increasing cost compared to DRO.

% \smallskip\noindent\textbf{Length-prediction baseline.} We add LP$+k\sigma$ baselines that reserve $r_i = L_{\text{in}}^i + \hat\mu_i + k\sigma_i$ for $k \in \{1, 2\}$, where $\hat\mu_i$ and $\sigma_i$ are the per-class empirical mean and standard deviation of historical output lengths. This is a class-level surrogate for the calibrated-predictor deployment pattern of S$^3$~\cite{jin2023s3} and Response-Length Perception~\cite{zheng2024response}; the safety-margin mechanism, not the predictor sophistication, is what governs cross-regime behavior. At $\rho = 10$, LP$+1\sigma$ achieves the lowest P99 ($11.30$\,s) and SLO violation ($13.2\%$) in the table, slightly better than DRO and P90---confirming that a calibrated predictor is competitive when $\rho$ is known. However, the safety margin scales with $\sigma_i$, not with the cost ratio. At $\rho = 100$, LP$+1\sigma$ under-reserves and costs $19{,}055$ ($4.0{\times}$ DRO), inheriting the same regime-narrowness as fixed-quantile reservations.

\begin{table}[t]
\centering
\caption{Cost across three cost-ratio regimes ($\rho \in \{2, 10, 100\}$) and latency at $\rho{=}10$ for fixed-quantile, length-prediction, and DRO baselines.}
\label{tab:latency_slo}
\small
\setlength{\tabcolsep}{3pt}
\begin{tabular}{l|ccc|cc}
\toprule
\textbf{Method} & \multicolumn{3}{c|}{\textbf{Cost}} & \multicolumn{2}{c}{\textbf{Latency at $\rho{=}10$}} \\
                & $\rho{=}2$ & $\rho{=}10$ & $\rho{=}100$ & \textbf{P99 (s)} & \textbf{SLO viol.} \\
\midrule
\textbf{DRO (ours)} & \textbf{1{,}270} & \textbf{3{,}083} & \textbf{4{,}807} & 11.47 & 15.6\% \\
P90              & 1{,}852 & 3{,}095 & 17{,}080 & 11.39 & 14.8\% \\
P95              & 2{,}890 & 3{,}320 & 8{,}153 & 12.50 & 23.2\% \\
P99              & 4{,}308 & 4{,}349 & 4{,}811 & 17.12 & 33.7\% \\
Mean             & 1{,}275 & 4{,}666 & 42{,}811 & 12.03 & 21.7\% \\
Max              & 4{,}752 & 4{,}756 & 4{,}807 & 19.04 & 36.7\% \\
\midrule
LP$+1\sigma$~\cite{jin2023s3,zheng2024response} & 1{,}729 & 3{,}143 & 19{,}055 & \textbf{11.30} & \textbf{13.2\%} \\
LP$+2\sigma$~\cite{jin2023s3,zheng2024response} & 2{,}615 & 3{,}207 & 9{,}865  & 11.80 & 18.0\% \\
\bottomrule
\end{tabular}
\end{table}

\noindent \textbf{\textit{Key Takeaway 4:}} DRO is the only method whose cost remains within a small constant factor of the best baseline across all $\rho$ regimes. At individual $\rho$ values, a regime-matched fixed quantile or LP$+k\sigma$ may marginally outperform DRO on latency, but each of these baselines is also several-fold worse than DRO at some other $\rho$. DRO trades small per-regime overhead for cross-regime stability.

\begin{table*}[t]
\centering
\caption{Ablation on a heterogeneous instance ($\rho{=}20$, load $\alpha_\lambda{=}0.6$, $2.5{\times}$ shifted evaluation; mean~$\pm$ std over $5$ seeds).}
\label{tab:ablation}
\footnotesize
\setlength{\tabcolsep}{3pt}
\begin{tabular}{l|c|c|c|c}
\toprule
\textbf{Variant} & \textbf{Cost} & \textbf{P99 (s)} & \textbf{Goodput} & \textbf{SLO viol.(\%)} \\
\midrule
\textbf{BCD-DRO (ours)}                & $\mathbf{13{,}476 \pm 208}$  & $\mathbf{8.12 \pm 0.65}$    & $\mathbf{5.01 \pm 0.25}$ & $\mathbf{26.4 \pm 3.7}$ \\
\,\,$-$ Routing (uniform $\pi$)        & $93{,}878 \pm 5{,}152$       & $177.62 \pm 11.90$          & $0.63 \pm 0.06$          & $93.8 \pm 0.6$ \\
\,\,\rev{$-$ DRO (SAA)}                     & $24{,}864 \pm 1{,}534$       & $11.02 \pm 1.55$           & $3.68 \pm 0.35$          & $45.8 \pm 5.2$ \\
\,\,$-$ Prefix caching ($\delta{=}0$)  & $13{,}476 \pm 208$           & $9.95 \pm 1.29$             & $4.39 \pm 0.37$          & $35.4 \pm 5.4$ \\
\bottomrule
\end{tabular}
\end{table*}

\subsection{Ablation Analysis}
\label{sec:ablation}
To evaluate the contribution of each component in BCD-DRO, we disable them one at a time and re-solve the optimization, then test each variant on a $2.5{\times}$ output-length-shifted distribution. The shifted setting reflects realistic training--evaluation mismatch and highlights the benefit of Wasserstein robustness over sample-average optimization. The instance forces routing diversity because the Code class exceeds the KV capacity of the small TP$=2$ configuration. Table~\ref{tab:ablation} reports the results. \rev{The $-$DRO gap grows monotonically with shift magnitude, as Fig.~\ref{fig:ablation_shift} shows across $\{1.0,\,1.5,\,2.0,\,2.5,\,3.0\}{\times}$.}

\noindent\textbf{Routing.} Uniform $\pi$ assigns Code requests to GPU groups whose KV capacity cannot accommodate them, causing request rejection or KV recomputation and eventually saturating the queue. On workloads with heterogeneous memory footprints, routing is a feasibility constraint rather than a performance-tuning parameter.

\noindent\textbf{Distributional robustness.} Setting $\varepsilon = 0$ reduces the model to the empirical quantile and only protects against tail events observed in the training data. Under distribution shift, the reserved capacity underestimates the realized tail demand, leading to increased preemption. As a result, SAA achieves slightly lower cost on the training trace but degrades significantly on shifted workloads. Fig.~\ref{fig:ablation_shift} shows this gap widens with shift magnitude: DRO's P99 grows mildly from $3.0$ to $9.1$\,s as the shift factor moves from $1.0{\times}$ to $3.0{\times}$, whereas SAA degrades faster, from $5.4$ to $18.0$\,s---roughly twice DRO's tail at the largest shift, with its SLO-violation rate climbing to $67.7\%$ versus DRO's $35.6\%$.

\noindent\textbf{Prefix caching.} Disabling $\delta_{i,k}$ leaves the reported cost unchanged because $\delta$ enters only the queueing/capacity model, not the preemption/waste cost objective; its effect appears entirely in queueing delay: per-request KV usage increases from $r_i - L_{\text{pre}}^i$ to $r_i$, increasing memory utilization and therefore waiting time under the Pollaczek--Khinchin model. The benefit of prefix caching is therefore driven primarily by concurrency rather than aggregate arrival rate.

\begin{figure}[t]
\centering
\includegraphics[width=0.85\columnwidth]{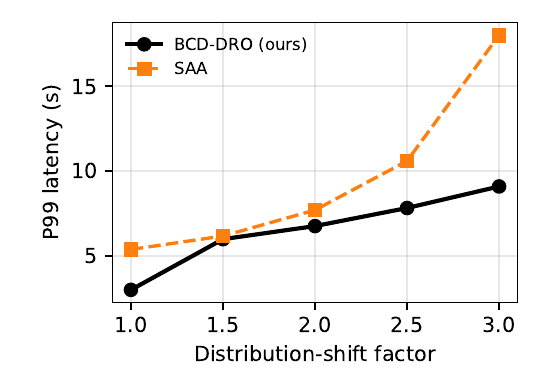}
\caption{BCD-DRO vs.\ SAA. The $x$-axis is the multiplicative shift factor $s$ (ratio of eval-distribution mean to training-distribution mean; $s{=}1$ means no shift). \rev{SAA's P99 stays consistently above DRO's and the gap widens with $s$, reaching $\approx\!18$\,s for SAA versus $\approx\!9$\,s for DRO at $s{=}3.0$.}}
\label{fig:ablation_shift}
\end{figure}

\noindent \textbf{\textit{Key Takeaway 5:}} Routing maintains feasibility under heterogeneous memory demands, the Wasserstein radius $\varepsilon$ improves robustness to workload drift, and prefix caching reduces queueing delay under high concurrency. The unified formulation manages these effects jointly with only two external parameters: $\varepsilon$ and the cost ratio $\rho$.

%==============================================================================
\section{Conclusion}
%==============================================================================

% We presented a robust KV cache management framework for LLM serving under output token length uncertainty. The proposed framework jointly optimizes GPU parallelism configuration, per-class KV cache reservation, request routing, and prefix caching under latency SLO constraints, enabling coordinated control of memory allocation, throughput, and tail latency in heterogeneous GPU clusters.
%We formulated a joint optimization for LLM serving that co-designs parallelism configuration, KV cache reservation, and request routing under output length uncertainty. 
We presented a robust KV cache management framework for LLM serving under output token length uncertainty. The framework jointly optimizes GPU parallelism configuration, KV cache reservation, request routing, and prefix caching under latency SLO constraints, enabling coordinated control of memory allocation, throughput, and tail latency in heterogeneous GPU clusters. A Wasserstein DRO formulation provides protection against distribution shifts, while the scalable BCD-DRO algorithm enables periodic re-optimization as workloads evolve.
%The distributionally robust formulation with Wasserstein ambiguity sets provides protection against distribution shifts, while the BCD-DRO algorithm scales to problems where direct MIBLP solving times out, enabling periodic re-optimization as workloads evolve. 
The critical fractile structure (Proposition~\ref{prop:critical_fractile}) automatically determines optimal reservation quantiles without manual heuristic tuning. %, eliminating the need for manual tuning of fixed heuristics.
Overall, this work highlights the importance of uncertainty-aware memory management in large-scale LLM serving. Future work includes integration with production serving runtimes, online adaptation using streaming workload statistics, and joint optimization with energy-aware and geographically distributed inference scheduling.
%This work highlights the importance of uncertainty-aware memory management in large-scale LLM serving. Future directions include integration with real serving runtimes such as vLLM and TensorRT-LLM, online adaptation using streaming workload statistics, and joint optimization of KV cache management with energy-aware and geographically distributed inference scheduling.

%==============================================================================

\appendix
\section{Kantorovich Duality Reformulation}
\label{app:duality}
%==============================================================================

This appendix provides a complete derivation of the tractable reformulation \eqref{eq:tractable}. 

\begin{proposition}[Worst-Case Expectation Reformulation]
\label{prop:duality}
The worst-case expectation over the Wasserstein ambiguity set $\sup_{P \in \calP_\varepsilon} \E_P[Q^{\text{stoch}}_i]$ equals the optimal value of a finite linear program with $O(N)$ variables and constraints.
\end{proposition}

\begin{proof}
The proof proceeds in five steps.

\smallskip\noindent\textbf{Step 1: Problem Setup.}
Consider the worst-case expectation over the Wasserstein ambiguity set:
\begin{equation}
    \sup_{P \in \calP_\varepsilon} \E_P[Q^{\text{stoch}}_i(r_i, \xi_i)]
    \label{eq:app_dro}
\end{equation}
where the ambiguity set is $\calP_\varepsilon = \{P : W_1(P, \hat{P}_N) \leq \varepsilon\}$. The 1-Wasserstein distance between distributions $P$ and $Q$ is defined as:
\begin{equation}
    W_1(P, Q) = \inf_{\pi \in \Pi(P, Q)} \int_{\Xi \times \Xi} |\xi - \xi'| \, d\pi(\xi, \xi')
    \label{eq:wasserstein_def}
\end{equation}
where $\Pi(P, Q)$ denotes the set of all joint distributions (couplings) with marginals $P$ and $Q$. The empirical distribution is $\hat{P}_N = \frac{1}{N}\sum_{n=1}^N \delta_{\hat{\xi}_i^{(n)}}$, and the stochastic cost is:
\begin{equation}
    Q^{\text{stoch}}_i(r_i, \xi) = c_p(L_{\text{in}}^i + \xi - r_i)^+ + c_w(r_i - L_{\text{in}}^i - \xi)^+
\end{equation}
The support of $\xi$ is $\Xi = [0, L_{\text{max}}^i]$.

\smallskip\noindent\textbf{Step 2: Lagrangian Relaxation and Strong Duality.}
The constraint $W_1(P, \hat{P}_N) \leq \varepsilon$ requires that the adversarial distribution $P$ lies within the Wasserstein ball. To handle this infinite-dimensional constraint, we introduce dual variable $\gamma_i \geq 0$ and form the Lagrangian:
\begin{equation}
    \mathcal{L}(P, \gamma_i) = \E_P[Q^{\text{stoch}}_i] - \gamma_i \left( W_1(P, \hat{P}_N) - \varepsilon \right)
\end{equation}
By strong duality~\cite{mohajerin2018data}:
\begin{equation}
    \sup_{P: W_1(P, \hat{P}_N) \leq \varepsilon} \E_P[Q^{\text{stoch}}_i] = \min_{\gamma_i \geq 0} \sup_{P} \mathcal{L}(P, \gamma_i)
    \label{eq:app_strong_duality}
\end{equation}

\smallskip\noindent\textbf{Step 3: Pointwise Evaluation via Kantorovich Duality.}
The inner supremum over $P$ in \eqref{eq:app_strong_duality} can be evaluated pointwise. For the empirical distribution $\hat{P}_N$, the optimal adversarial strategy transports each sample $\hat{\xi}_i^{(n)}$ to some point $\xi_n \in \Xi$, paying transportation cost $|\xi_n - \hat{\xi}_i^{(n)}|$. This yields the Kantorovich dual form:
\begin{equation}
    \!\! \min_{\gamma_i \geq 0} \left\{ \gamma_i \varepsilon \!+\! \frac{1}{N}\sum_{n=1}^N \sup_{\xi \in \Xi} \left\{ Q^{\text{stoch}}_i(r_i, \xi) \!-\! \gamma_i |\xi - \hat{\xi}_i^{(n)}| \right\} \right\}
    \label{eq:app_kantorovich}
\end{equation}
The dual variable $\gamma_i$ acts as a transportation cost penalty: larger $\gamma_i$ penalizes moving probability mass more heavily, keeping the worst-case distribution closer to $\hat{P}_N$.

\smallskip\noindent\textbf{Step 4: Closed-Form Pointwise Supremum.}
Define the adversarial cost function for each sample:
\begin{equation}
    \phi_i(r_i, \hat{\xi}, \gamma_i) = \sup_{\xi \in [0, L_{\text{max}}^i]} \left\{ Q^{\text{stoch}}_i(r_i, \xi) - \gamma_i |\xi - \hat{\xi}| \right\}
    \label{eq:app_phi}
\end{equation}
Since $Q^{\text{stoch}}_i(r_i, \xi)$ is piecewise linear in $\xi$ with breakpoint at $\xi^* = r_i - L_{\text{in}}^i$, and the transportation cost $|\xi - \hat{\xi}|$ is piecewise linear with breakpoint at $\hat{\xi}$, the objective inside the supremum is piecewise linear. For piecewise linear functions, the supremum over a compact interval is attained at either a vertex of the domain or a breakpoint. The candidate points are $\xi = 0$ (lower boundary), $\xi = \hat{\xi}$ (sample point), and $\xi = L_{\text{max}}^i$ (upper boundary). Evaluating at each:
\begin{align}
    \phi_i^{(0)} &= c_w(r_i - L_{\text{in}}^i)^+ - \gamma_i \hat{\xi} \label{eq:phi_0} \\
    \phi_i^{(\hat{\xi})} &= c_p(\hat{\xi} + L_{\text{in}}^i - r_i)^+ + c_w(r_i - L_{\text{in}}^i - \hat{\xi})^+ \label{eq:phi_sample} \\
    \phi_i^{(L)} &= c_p(L_{\text{max}}^i + L_{\text{in}}^i - r_i) - \gamma_i(L_{\text{max}}^i - \hat{\xi}) \label{eq:phi_max}
\end{align}
The closed-form supremum is $\phi_i = \max\{\phi_i^{(0)}, \phi_i^{(\hat{\xi})}, \phi_i^{(L)}\}$.

\smallskip\noindent\textbf{Step 5: Epigraph Reformulation.}
Substituting the closed-form $\phi_i$ back into \eqref{eq:app_kantorovich}:
\begin{equation}
    \min_{\gamma_i \geq 0} \left\{ \gamma_i \varepsilon + \frac{1}{N}\sum_{n=1}^N \phi_i(r_i, \hat{\xi}_i^{(n)}, \gamma_i) \right\}
\end{equation}
The $\max$ inside $\phi_i$ is non-smooth. We introduce auxiliary variables $\theta_{i,n}$ and use epigraph representation $\theta_{i,n} \geq \phi_i(r_i, \hat{\xi}_i^{(n)}, \gamma_i)$, which is equivalent to enforcing $\theta_{i,n} \geq \phi_i^{(0)}$, $\theta_{i,n} \geq \phi_i^{(\hat{\xi})}$, and $\theta_{i,n} \geq \phi_i^{(L)}$. The final linear program is:
\begin{subequations}
\label{eq:app_lp}
\begin{align}
    & \min_{\gamma_i, \theta} ~ \gamma_i \varepsilon + \frac{1}{N}\sum_{n=1}^N \theta_{i,n} \label{eq:app_obj}\\
    & \text{s.t.} ~  \theta_{i,n} \geq c_p(\hat{\xi}_i^{(n)} + L_{\text{in}}^i - r_i), ~ \forall n \label{eq:app_c1}\\
    & \theta_{i,n} \geq c_w(r_i - L_{\text{in}}^i - \hat{\xi}_i^{(n)}), ~ \forall n \label{eq:app_c2}\\
    & \theta_{i,n} \geq c_p(L_{\text{max}}^i + L_{\text{in}}^i - r_i) - \gamma_i(L_{\text{max}}^i - \hat{\xi}_i^{(n)}), ~ \forall n \label{eq:app_c3}\\
    & \theta_{i,n} \geq c_w(r_i - L_{\text{in}}^i) - \gamma_i \hat{\xi}_i^{(n)}, ~ \forall n \label{eq:app_c4}\\
    & \gamma_i \geq 0, ~ \theta_{i,n} \geq 0, ~ \forall n \label{eq:app_c5}
\end{align}
\end{subequations}
This is a linear program in $(\gamma_i, \theta)$ for fixed $r_i$. Constraints \eqref{eq:app_c1}--\eqref{eq:app_c4} correspond to \eqref{eq:epi_preempt}--\eqref{eq:epi_lower} in the main text.
\end{proof}

\newpage
\bibliographystyle{IEEEtran}

\begin{thebibliography}{10}

\bibitem{vllm2023}
W.~Kwon, Z.~Li, \emph{et al.}, ``Efficient memory management for large language model serving with PagedAttention,'' in \emph{Proc. SOSP}, 2023.

\bibitem{yu2022orca}
G.-I.~Yu, J.~S.~Jeong, \emph{et al.}, ``Orca: A distributed serving system for transformer-based generative models,'' in \emph{Proc. OSDI}, 2022.

\bibitem{agrawal2024taming}
A.~Agrawal, N.~Kedia, \emph{et al.}, ``Taming throughput-latency tradeoff in LLM inference with Sarathi-Serve,'' in \emph{Proc. OSDI}, 2024.

\bibitem{zhang2023h2o}
Z.~Zhang, Y.~Sheng, \emph{et al.}, ``H2O: Heavy-hitter oracle for efficient generative inference of large language models,'' in \emph{Proc. NeurIPS}, 2023.

\bibitem{li2024snapkv}
Y.~Li, Y.~Huang, \emph{et al.}, ``SnapKV: LLM knows what you are looking for before generation,'' \emph{arXiv:2404.14469}, 2024.


\bibitem{jin2023s3}
Y.~Jin, C.-F.~Wu, D.~Brooks, and G.-Y.~Wei, ``S$^3$: Increasing GPU utilization during generative inference for higher throughput,'' in \emph{Proc. NeurIPS}, 2023.

\bibitem{zheng2024response}
Z.~Zheng, X.~Ren, \emph{et al.}, ``Response length perception and sequence scheduling: An LLM-empowered LLM inference pipeline,'' in \emph{Proc. NeurIPS}, 2024.

\bibitem{qiu2024power}
H.~Qiu, W.~Mao, \emph{et al.}, ``Power-aware deep learning model serving with $\mu$-Serve,'' in \emph{Proc. USENIX ATC}, 2024.


\bibitem{shoeybi2019megatron}
M.~Shoeybi, M.~Patwary, \emph{et al.}, ``Megatron-LM: Training multi-billion parameter language models using model parallelism,'' \emph{arXiv:1909.08053}, 2019.

\bibitem{pope2023efficiently}
R.~Pope, S.~Douglas, \emph{et al.}, ``Efficiently scaling transformer inference,'' in \emph{Proc. MLSys}, 2023.

\bibitem{zheng2023judging}
L.~Zheng, W.-L.~Chiang, \emph{et al.}, ``Judging LLM-as-a-judge with MT-Bench and Chatbot Arena,'' in \emph{Proc. NeurIPS}, 2023.

\bibitem{cheng_fastllm}
J.~Cheng and D.~Nguyen, ``Scalable Joint Resource Allocation for SLO-Constrained LLM Inference in Heterogeneous GPU Clouds,'' \emph{arXiv:2504.07472}, 2025.

\bibitem{cheng_greenllm}
J.~Cheng and D.~Nguyen, ``Green-LLM: Optimal workload allocation for environmentally-aware distributed inference,'' \emph{arXiv:2507.09942}, 2025.

\bibitem{mohajerin2018data}
P.~Mohajerin~Esfahani and D.~Kuhn, ``Data-driven distributionally robust optimization using the Wasserstein metric,'' \emph{Math. Program.}, vol.~171, pp.~115--166, 2018.

\bibitem{edgeworth1888}
F.~Y.~Edgeworth, ``The mathematical theory of banking,'' \emph{J. Roy. Statist. Soc.}, vol.~51, no.~1, pp.~113--127, 1888.

\bibitem{arrow1951}
K.~J.~Arrow, T.~Harris, and J.~Marschak, ``Optimal inventory policy,'' \emph{Econometrica}, vol.~19, no.~3, pp.~250--272, 1951.

\bibitem{lee2021}
S.~Lee, H.~Kim, and I.~Moon, ``A data-driven distributionally robust newsvendor model with a Wasserstein ambiguity set,'' \emph{J. Oper. Res. Soc.}, vol.~72, no.~8, pp.~1879--1897, 2021.

\bibitem{kleinrock1975}
L.~Kleinrock, \emph{Queueing Systems, Volume 1: Theory}. Wiley, 1975.

\bibitem{queueing_llm2024}
Y.~Yang, Y.~Xu, and L.~Jiao, ``A queueing theoretic perspective on low-latency LLM inference with variable token length,'' in \emph{Proc. WiOpt}, 2024.

\bibitem{jager2019blockwise}
S.~J\"{a}ger and A.~Sch\"{o}bel, ``The blockwise coordinate descent method for integer programs,'' \emph{Math. Meth. Oper. Res.}, vol.~91, pp.~357--381, 2020.

\bibitem{wang2024burstgpt}
Y.~Wang, Y.~Chen, \emph{et al.}, ``BurstGPT: A real-world workload dataset to optimize LLM serving systems,'' in \emph{Proc. KDD}, 2025.


\bibitem{stojkovic2024dynamollm}
J.~Stojkovic, C.~Zhang, I.~Goiri, J.~Torrellas, and E.~Choukse, ``DynamoLLM: Designing LLM inference clusters for performance and energy efficiency,'' in \emph{Proc. HPCA}, 2025.

\bibitem{sharegpt2023}
ShareGPT Team, ``ShareGPT: Share your ChatGPT conversations,'' 2023. [Online]. Available: \url{https://sharegpt.com}

\bibitem{zhong2024distserve}
Y.~Zhong, S.~Liu, \emph{et al.}, ``DistServe: Disaggregating prefill and decoding for goodput-optimized large language model serving,'' in \emph{Proc. OSDI}, 2024.

\end{thebibliography}

\end{document}